\documentclass[10pt,a4paper]{article}
\usepackage{amsmath,amssymb,amsfonts,amsthm}
\usepackage{graphicx}
\usepackage{t1enc}
\usepackage[english]{babel}
\usepackage[latin1]{inputenc}
\usepackage{color}
\usepackage{verbatim}
\usepackage{hyperref}
\usepackage{color}
\usepackage{booktabs}
\newtheorem{theorem}{Theorem}[section]

\newtheorem{definition}[theorem]{Definition}

\theoremstyle{definition}

\usepackage{enumerate}

\newcommand{\be}{\begin{equation}}

\newcommand{\ee}{\end{equation}}

\newcommand{\Rt}{\mathbb{R}^3}

\newcommand{\RA}{\mathcal{R}_A}
\newcommand{\RC}{\mathcal{R}_C}

\newcommand{\dv}{dv}

\newcommand{\dvf}{dv_0}

\newcommand{\pd}{\partial }

\title{Size, angular momentum and mass for objects}
\author{Pablo Anglada, M.E. Gabach-Clement\footnote{gabach@famaf.unc.edu.ar}, Omar E. Ortiz\\
  Facultad de Matem\'atica, Astronom\'{i}a y F\'{i}sica, \\
     Universidad Nacional de C\'ordoba, \\
Instituto de F\'{i}sica Enrique Gaviola, IFEG, CONICET,\\
  Ciudad Universitaria (5000) C\'ordoba, Argentina.}

\begin{document}
\maketitle
\begin{abstract}
We obtain a geometrical inequality involving the  ADM mass, the angular momentum and the size of an  ordinary, axially symmetric object. 
We use the monotonicity of the Geroch 
quasi-local energy on 2-surfaces along the inverse mean curvature flow. We also compute numerical examples to test the robustness of our hypotheses and results
\end{abstract}

\section{Introduction}

During the last decade, geometrical inequalities for black holes have received
much attention and different relations involving the total mass, the angular
momentum, the horizon area, the electromagnetic charge, the cosmological
constant and certain shape parameters have been found
\cite{Dain06c,Hennig:2008zy,Acena:2010ws,Dain:2011pi,Clement:2012vb,Schoen:2012nh,Clement:2015fqa}.
The question of whether the same kind of relations hold for ordinary objects
(\textit{i.e.} not black holes) is not trivial. Black holes are very special
solutions of Einstein equations that can be described by few parameters,
at least in the stationary limit. The geometrical inequalities mentioned above
show that this no-hair property of the stationary state sets restrictions on the
values that physical quantities can have in the general, dynamical black hole
state.

On the other hand, ordinary objects like neutron stars are not simple, and 
hence finding such simple relations  is not \textit{a priori} expected. 
One of the first works in this direction is due to Schoen and Yau \cite{SchoenYau1983} (see the discussion and references in \cite{Reiris:2014tva}).
They found a lower bound to the Ricci scalar (and hence, to the matter density) in asymptotically flat initial data only in terms of a certain 
radius 
characterizing the object. This in turn gave rise to a black hole formation criteria due to concentration of matter. There were other results similar 
in nature to this one, namely, quasi-local inequalities between size and matter/currents density for objects, with alternative measures of size 
\cite{Murchadha86b} \cite{Galloway:2008gc}.
More recently, the Schoen-Yau bound was used by Dain 
\cite{Dain:2013gma} and Khuri \cite{Khuri:2015zla} under different conditions in axial symmetry, to find a quasi-local relation between angular momentum and size, of the form
$J\lesssim \mathcal R^2$. Furthermore Khuri 
\cite{Khuri:2015xpa} also used it to prove an inequality for charged objects 
using a similar measure of size. 
Independently, Reiris \cite{Reiris:2014tva} derived quasi-local inequalities 
relating the angular momentum and charge to precise measures of size and 
shape
of ordinary objects. Using a similar approach, Khuri \cite{Khuri:2016ulv} also
established inequalities relating size, local mass, angular momentum, and charge, that give rise to black hole existence criteria.

The   
appropriate measure of the size of an object is not easy to determine, 
nor is  finding its relation with relevant physical quantities, hence these 
inequalities are not expected to be sharp.
Moreover, it is expected that these inequalities become saturated  for very 
special cases, for 
example the inequality between charge and size, $2\mathcal R > Q$, is sharp  in 
the spherically symmetric case with $\mathcal R$ the areal radius of the object 
\cite{Anglada:2015tan}, and the equality is achieved for a sequence of objects 
whose charge, mass, and radius tend to zero.

Our interest in this article is to relate the rotation of an ordinary object (i.e. its angular momentum) with the total energy (the ADM mass, see \cite{Arnowitt62}) in axially symmetric systems. This 
approach is inspired by the slow rotation treatment of neutron stars. As discussed in \cite{haensel2007neutron}, rotating neutron stars are axially symmetric and for slow rotation, the rotational perturbations of the stellar structure are quadratic in the 
angular velocity and therefore, quadratic in the angular momentum. In the Newtonian limit, one may write the total energy of the star as the sum of 
two terms, the first including the gravitational and internal energies, 
denoted as $E_0$, and the second, the rotational energy
\be\label{Enewton}
E\approx E_0+\frac{J^2}{2I}
\ee
where $J$ is the angular momentum and $I$ the stellar moment of inertia. In 1967 Hartle \cite{Hartle:1967he} devised a perturbative method to compute the first order correction to the 
neutron star's energy in the context of general relativity, and found an 
expression similar to \eqref{Enewton}, where the term $E_0$ represents the 
total energy for the non-rotating star, and the quadratic term in the angular 
momentum is the rotational contribution to first order. An extensive study of 
this equation and different relations between the kinetic and binding energies
has been performed since the late 1960's (see \cite{Stergioulas:2003yp} for further details and references). 
It is not our aim to address these problems, but to seek geometrical relations between certain physical parameters for a rotating ordinary object, like the neutron star mentioned above.

Within Newtonian theory, there are no restrictions to the values that the quantities in \eqref{Enewton} can attain (we are not considering mechanical processes related to
the particular equations of state of matter). Nevertheless, in general relativity, the Hoop conjecture \cite{Senovilla:2007dw} sets bounds on the (quasi-local) mass that an object can have. 
This conjecture roughly says that if one is able to pass a hoop of 
 radius $\mathcal R$ in every direction around a region $\Omega$ with (quasi-local) mass $m_{\Omega}$, then it will collapse to form a black hole if 
\be\label{hoop}
m_{\Omega}>\frac{\mathcal R}{2}.
\ee
Otherwise, if \eqref{hoop} is not satisfied, one expects an ordinary object.

To incorporate this expected constraint into \eqref{Enewton}, we start with the Newtonian definition of moment of inertia of $\Omega$ with respect to the symmetry axis
\be
I:=\int_\Omega \mu\rho^2dV
\ee
where $\mu$ is the matter density and $\rho$ is the (euclidean) distance from the axis. Moreover, we can roughly write
\be\label{inertiaapprox}
I\approx m_\Omega\mathcal R_C^2
\ee
where $m_\Omega$ is some measure of the (quasi local) mass contained in the object and $\mathcal R_C=\max \rho$ is the circumferential radius. For instance, for a spheroid rotating around the axis of (semi-)length $\mathcal R_C$, the moment of inertia is
$2/5m_{\Omega}\mathcal R_C^2$. Then, putting together \eqref{hoop} and \eqref{inertiaapprox} we may estimate the following bound to the 
total energy of a \textit{non-collapsing} object
\be\label{Enewton2}
E\gtrsim E_0+\frac{J^2}{\mathcal R\mathcal R_C^2}.
\ee
Although naive and informal, this expression gives a lower bound on the total energy for a rotating system in terms of the angular momentum and two measures of size, that we call $\mathcal R$
and $\mathcal R_C$. It is interesting to note that these two quantities come from different contexts. The \textit{distance} to the rotation axis at which matter is located, 
represented in our argument by the circumferential radius $\mathcal R_C$, seems 
to be (at least from Newtonian experience) the relevant quantity to describe the 
kinetic 
rotational energy. In other words, for rotation it is important to account for the spread out of matter with respect to the symmetry/rotation axis. On the other hand, the measure $\mathcal R$ coming from the Hoop conjecture should describe all directions in which matter is spread out. This quantity
cares about how localized in every direction matter is.

In this article we present a relation similar to \eqref{Enewton2} for an ordinary, isolated, rotating and axially symmetric object in  general relativity.

The main tools we use to accomplish this are the inverse mean curvature flow
(IMCF) and the Geroch energy, which  have proven to be useful in obtaining
geometrical inequalities in general relativity like the Riemannian Penrose
inequality \cite{Huisken01}.  More related to our system, Dain used the IMCF and
the Hawking energy to obtain an inequality between size, ADM mass and electric
charge for ordinary objects \cite{Dain:2015ira}. Our aim in this article is to
explore what the IMCF can say about rotating ordinary objects. The main
difficulty being the explicit inclusion of the angular momentum into the
geometrical relations.

The article is organized as follows: In section \ref{sec:IMCF} we review the
inverse mean curvature flow and some important properties we will use in
deriving our result. Also we introduce the Geroch energy and discuss the
monotonicity properties. In section \ref{sec:Main} we present our hypotheses,
the main theorem and discussions.  In section \ref{sec:num} we show numerical
solutions of the flow equations where we test both some of  the hypotheses of
our main result, and the geometrical inequality we have found for objects.

\section{IMCF and Geroch energy}\label{sec:IMCF}                                                                                                   
In this section we review the basic properties of the inverse mean curvature flow (IMCF) and the Geroch energy \cite{Huiskenevol}, \cite{Szabados04}.                        

Consider a smooth Riemannian 3-manifold  $M$ with metric $\bar g_{ij}$, connection $\bar \nabla_i$ and Ricci curvature $\bar R_{ij}$. A  solution of the IMCF is a 
smooth family of hypersurfaces  $S_t:=x(S,t)$ with $x:S\times[0,\tau]\to M$ satisfying the evolution equation
\be \label{eqIMCF}
\frac{\partial x}{\partial t}=\frac{\nu}{H}
\ee
where $t\in[0,\tau]$, $H>0$ is the mean curvature of the 2-surface $S_t$ at $x$ and $\nu$ is the outward unit normal to $S_t$.

Let $g_{ij}$ be the induced metric on $S_t$, $\nabla_i$ the covariant 
derivative, $h_{ij}$ the second fundamental form and $ds$ the area element of 
$S_t$. Then one can derive the 
evolution equations (see \cite{Huiskenevol} more details)
\be\label{evolg}
\frac{\partial}{\partial t}g_{ij}=\frac{2}{H}h_{ij}
\ee
\be\label{evolarea}
\frac{\partial}{\partial t}(ds)= ds
\ee
\be\label{evolH}
\frac{\partial}{\partial t}H=-\Delta(H^{-1})-H^{-1}(|h|^2+\bar R_{ij}\nu^i\nu^j).
\ee

It is important to note that when considering the IMCF in axially symmetric initial data, 
the IMCF equation \eqref{eqIMCF} preserves axial symmetry. That is, if one starts the flow out of a point on the symmetry axis (or out of an axially 
symmetric initial 
surface) there is no mechanism that could make the normal to each subsequent surface to have a component along the axial Killing vector. Due to this 
observation, from now on, when we discuss the IMCF flow, we always consider it consisting of axially symmetric surfaces $S_t$.

On each surface $S_t$ we introduce the Geroch energy
\be \label{Gmass}
E_G(S_t):=\frac{A_t^{1/2}}{(16\pi)^{3/2}}\left(16\pi-\int_{S_t}H^2ds\right)
\ee
where $A_t$ is the area of $S_t$. This functional has some interesting properties that will be used later in the proof of our result. Namely, for a complete, maximal, asymptotically flat initial
data, with non-negative scalar curvature, and surfaces $S_t$ that are topological spheres converging to round spheres at infinity, $E_G$ satisfies 
\be\label{propEG}
E_G(S_t)\geq0,\qquad \frac{dE_G}{dt}(S_t)\geq0,\qquad  \lim_{t \to \infty} E_G(S_t) = m_{ADM}.
\ee
We refer the reader to \cite{Huisken01} for details, proofs and further references. However, since it will be relevant in proving our main theorem, we will sketch the proof of the 
monotonicity property. 

We start with the time derivative of $E_G(S_t)$
\begin{multline}
\frac{d}{dt}E_G=\frac{A_t^{1/2}}{(16\pi)^{3/2}}\left[8\pi- 
\frac{1}{2}\int_{S_t}H^2 ds + \right.\\
\left. \int_{S_t}\left(-H^2+2H\Delta(H^{-1})+2|h|^2+2\bar 
R_{ij}\nu^i\nu^j\right)ds\right]
\end{multline}
where we have used \eqref{evolarea} and \eqref{evolH}. Then we use the Gauss equation
\be
2\bar R_{ij}\nu^i\nu^j=\bar R+H^2-|h|^2-2\kappa
\ee
where $\kappa$ is the Gauss curvature, and obtain
\be
\frac{d}{dt}E_G=\frac{A_t^{1/2}}{(16\pi)^{3/2}}\left[8\pi +\int_{S_t} 
\left(2H\Delta(H^{-1})+|h|^2 +\bar R -2\kappa - \frac{H^2}{2}\right)ds\right].
\ee
 Rewriting the $|h|^2$ term in terms of the principal curvatures $\lambda_1$ and 
$\lambda_2$ and the mean curvature we find
\be
\frac{d}{dt}E_G=\frac{A_t^{1/2}}{(16\pi)^{3/2}}\left[8\pi +\int_{S_t} 
\left(2H\Delta(H^{-1})+\frac{1}{2}(\lambda_1-\lambda_2)^2 +\bar 
R-2\kappa\right)ds\right].
\ee
Next, we use the Gauss-Bonnet theorem and integrate by parts the Laplace operator
\be
\frac{d}{dt}E_G=\frac{A_t^{1/2}}{(16\pi)^{3/2}}\left[8\pi-4\pi\chi(S_t)+\int_{S_t
}\left(2\frac{|\nabla H|^2}{H^2}+\frac{1}{2}(\lambda_1-\lambda_2)^2+\bar 
R\right)ds\right]
\ee
where $\chi(S_t)$ is the surface's Euler characteristic.
If the surfaces $S_t$ are topological spheres  we have
\be\label{eqvolmm}
\frac{d}{dt}E_G=\frac{A_t^{1/2}}{(16\pi)^{3/2}}\int_{S_t}\left(2\frac{|\nabla 
H|^2}{H^2}+\frac{1}{2}(\lambda_1-\lambda_2)^2+\bar R\right)ds
\ee
from where we see that if the 3-manifold has nonnegative scalar curvature, then the Geroch energy is non-decreasing, that is $dE_G/dt\geq0$.

\section{Main result}\label{sec:Main}

Following \cite{Malec:2002ki}, we consider a complete initial data $(M,\bar g, K; \mu, j)$,
where $K$ is the extrinsic curvature of the 3-manifold $M$, and $\mu$, $j$ are the matter density and the mater current density respectively. 
We take this initial data to be maximal, asymptotically flat and axially symmetric, and we assume it satisfies the Dominant Energy Condition 
(DEC), $\mu\geq|j|$. Since we are concerned with ordinary objects, as opposed to black holes, we also require the initial data to have no 
minimal surfaces. 
The definition of a rotating object we will use is the following:

\begin{definition}{Object:}
Open set $\Omega$ in $M$ which is axially symmetric, compact, connected and such that the matter current density $j$ has compact support in $\Omega$.
\end{definition}

Assume that on $(M,\bar g)$ there exists a smooth inverse mean curvature flow of compact surfaces $S_t$, having spherical topology and going to 
 round spheres at infinity.  
Our aim is to relate 
the region $\Omega$ with the surfaces given by the IMCF to obtain a geometrical inequality involving physical parameters of $\Omega$.
In order to do this, we will take into account the fact that the asymptotic behavior of the surfaces $S_t$
implies that after some time $T$, ${S_t }$ will be convex. And also, the fact that maximality of  the initial data, together with the DEC imply, via the constraint
equations, that  
$\bar R \geq 0$, and hence the Geroch energy is non-decreasing. This will be crucial in what follows. It is also important to note that the assumption about
the smoothness of the flow could be relaxed. That is, some parts of our derivation do not require smoothness, and can be done using the weak 
level set version of the flow defined by Huisken and Ilmanen \cite{Huisken01}. Nevertheless, for simplicity of presentation we consider only the smooth case in 
this article.

Besides the ADM mass, $m_{ADM}$, the physical and geometrical quantities we are interested in are  the Komar angular momentum $J(S)$
and the areal and circumferential radii of a surface $S$ in $M$:

\be\label{angmom}
 J(S)=\frac{1}{8\pi}\int_{S} K_{ij} \eta^i \nu^jds,
\ee
\be\label{size}
\RA(S):=\sqrt{\frac{A}{4\pi}},\qquad \RC(S):=\frac{\mathcal C}{2\pi}
\ee
where $\eta^i$ is the Killing vector field associated to the axial symmetry,
$A$ is the area of $S$ and $\mathcal C$ is the length of the greatest axisymmetric circle of $S$. 

We find the following result.

\begin{theorem}
Let $(M, \bar g,K; \mu, j)$
be a maximal, asymptotically flat, axially symmetric initial data, that contains an object $\Omega$; the data satisfies the 
dominant energy condition and has no minimal surfaces. 
Assume there exists a smooth IMCF of surfaces $S_t$ on $M$ starting from a point on the symmetry axis inside $\Omega$ and such that $S_t$ is convex 
for $t\geq T\in \mathbb R$ and $S_T$ encloses the object. Then
\be\label{mainineq}
m_{ADM}\geq m_{T} +\frac{1}{5}\frac{J^2}{\RA \RC^2}
\ee
where $J$, $\RA$ and $\RC$ are the angular momentum, areal radius and circumferential radius of $S_T$ respectively, and 
\be\label{mT}
m_{T}:=\frac{1}{16\pi}\int_0^{\RA}d\xi\int_{S_\xi}\bar R ds 
\ee
and $\xi$ stands for the  areal radius coordinate.
\end{theorem}

\begin{proof}
The scheme of the proof is to start with the time derivative of the Geroch energy \eqref{eqvolmm}, bound away the $|\nabla H|^2$ and 
$(\lambda_1-\lambda_2)^2$ terms and use the constraint equations to write $\bar R$ in terms of the angular momentum. Then integrate in the flow parameter $t$ to infinity.

With this in mind, we have the bound (for simplicity we omit the area element 
$ds$ when possible)
\be\label{evol2}
\frac{d}{dt}E_G\geq\frac{A_t^{1/2}}{(16\pi)^{3/2}}\int_{S_t} \left[16 \pi \mu +  K_{ij}  K^{ij} \right]
\ee
where we have used the constraint $\bar R=16 \pi \mu +  K_{ij}  K^{ij}-(\mbox{tr}K)^2$ and maximality (i.e. tr $K=0$). In order to include the angular momentum into the inequality, we 
use  the Cauchy-Schwarz inequality in the definition of 
$J_t:=J(S_t)$  
\be
\begin{split}
  J_t^2 &= \left( \frac{1}{8\pi} \left|  \int_{S_t} K_{ij} \eta^i \nu^j \right| \right)^2 
  \leq \frac{1}{(8\pi)^2} \left( \int_{S_t} \mid  K_{ij} \eta^i \nu^j \mid \right)^2 \\
             &\leq \frac{1}{(8\pi)^2} \left( \int_{S_t} \mid  K_{ij} \mid  \sqrt{\eta} \right)^2\leq \frac{1}{(8\pi)^2} 
             \int_{S_t} \mid  K_{ij} \mid ^2  \int_{S_t} \eta  
\end{split}
\ee
where $\eta:=\eta_i\eta^i$ is the square norm of $\eta^i$ and in the fourth step we have used the H\"older inequality with $p=q=2$. Hence, we have a 
bound for the angular momentum of $S_t$  in terms of the extrinsic curvature (and hence, of the scalar curvature):
\be \label{JvsK}
\int_{S_t}  K_{ij}K^{ij}   \geq (8\pi)^2 \frac{ J_t^2}{\int_{S_t} \eta}. 
\ee
Putting this into  \eqref{evol2} we get
\be\label{evol3}
\frac{d}{dt}E_G\geq\frac{A_t^{1/2}}{(16\pi)^{3/2}}\left[\int_{S_t} 16 \pi \mu +  (8\pi)^2 \frac{ J_t^2}{\int_{S_t} \eta}  \right].
\ee
Now, let $T$ be the smallest time such that for $t\geq T$ the surfaces given by the flow are convex and such that $S_T$ encloses the object we are 
interested in.

Assuming there are no minimal surfaces in the space, the flow goes to infinity.
We integrate equation (22) from the initial time to $T$ and then to infinity. We first show that the integral of (22) from $0$ to $T$ is just the 
 quasi-local mass $m_{T}$. We start with the scalar constraint equation 
$\bar R=16 \pi \mu +  K_{ij}  K^{ij}$ and change the coordinate $t$ to the areal radius coordinate $\xi(t)=\sqrt{\frac{A_t}{4\pi}}$. We have $
 T  \rightarrow \xi(T)=\sqrt{\frac{A_T}{4\pi}}=\RA$
and $ dt \rightarrow d\xi = \sqrt{\frac{A_t}{16\pi}}dt$, thus:
\be\label{mT}
\int_0^T \frac{A_t^{1/2}}{(16\pi)^{3/2}}\int_{S_t}\bar R dsdt=\frac{1}{16\pi}\int_0^{\RA}d\xi\int_{S_\xi}\bar R ds =m_{T}
\ee
Then for the range $[T,\infty)$  we use equation (25) instead of (22) to explicitely include the angular momentum.
Due to the compact support of $j$,  $J_t$ is conserved outside $S_T$. Therefore $J_t=J_T:=J$, and hence, disregarding the positive term involving
$\mu$
and using the relation between the Geroch energy and the ADM mass at infinity, we obtain
\be\label{evol4}
m_{ADM}\geq\lim_{t \to \infty} E_G(S_t) \geq m_T + \sqrt{\pi}J^2\int_T^\infty \frac{A_t^{1/2}}{\int_{S_t}  \eta }dt.
\ee

Next we need to bound the surface integral of $\eta$. Here it is where convexity plays a role.
We introduce orthogonal coordinates $\theta, \varphi$ for the surface $S_t$ such that $\eta^i=\left(\frac{\partial}{\partial \varphi}\right)^i$. 
One can always do this for axially symmetric 2-surfaces that are diffeomorphic to $S^2$, see for example
\cite{Dain:2011pi}. Then we write the evolution equation \eqref{evolg} in the form\\
\be\label{evoleta}
\frac{\partial}{\partial t}\eta=\frac{\partial}{\partial t}g_{\varphi\varphi}=\frac{2}{H}h_{\varphi\varphi}.
\ee
Recall that in axial symmetry the principal and mean curvatures are given by
\be
\lambda_1=g^{\theta\theta}h_{\theta\theta},\qquad \lambda_2=g^{\varphi\varphi}h_{\varphi\varphi}
\ee
\be
H=g^{ij}h_{ij}=g^{\theta\theta}h_{\theta\theta}+g^{\varphi\varphi}h_{\varphi\varphi}=\lambda_1+\lambda_2.
\ee
Therefore we have
\be
h_{\varphi\varphi}=\frac{\lambda_2}{g^{\varphi\varphi}}=g_{\varphi\varphi}\lambda_2=\eta\lambda_2.
\ee
Putting this into \eqref{evoleta} we find
\be
\frac{\partial}{\partial t}\eta=2\eta\frac{\lambda_2}{H}. 
\ee
Now, we use this equation for $t\geq T$, where the surfaces are strictly convex, and therefore $\lambda_1,\lambda_2>0$ and $\lambda_2\leq H$. This gives us
\be\label{evoleta3}
\frac{\partial}{\partial t}\eta\leq2\eta
\ee
and
\be
\frac{\partial}{\partial t}(\eta dS)=\frac{\partial\eta}{\partial t} dS+\eta dS\leq2\eta dS+\eta dS=3\eta dS.
\ee
Therefore we can write $\eta dS\leq\eta_T e^{3(t-T)}\,dS_T $ and have 
\be
\int_{S_t}\eta\leq e^{3(t-T)}\int_{S_T}\eta_T\leq e^{3(t-T)}A_T\max_{S_T}\eta ,\qquad t\geq T.
\ee
Also, using $A_t=e^{(t-T)}A_T$ we bound
\be
\begin{split}
m_{ADM}\geq& \; m_T + \sqrt{\pi}J^2\int_T^\infty \frac{A_T^{1/2}e^{(t-T)/2}}{e^{3(t-T)}A_T\max_{S_T}\eta }dt\\
\geq& \;m_T+\frac{2\sqrt{\pi}}{5}\frac{J^2}{A_T^{1/2}\max_{S_T}\eta  }
\end{split}
\ee
where $m_T$, given by \eqref{mT} comes from integrating \eqref{evol2} in the interval $t\in (0,T)$.

Finally we write this expression in terms of the areal and circumferential radii \eqref{size}
and obtain \eqref{mainineq}.

\end{proof}

\textbf{Remarks}

\vspace{0.5cm}

 The inequality \eqref{mainineq} is global in nature as it involves the ADM mass. This is different from the quasi-local inequalities mentioned in the 
 Introduction. We also note that the Geroch energy seems to be a very appropriate quasi-local mass for our purposes. On one hand, it converges to the ADM mass. And
 on the other hand, it is directly related  to the mean curvature of the surface, and therefore, its time derivative along the flow is related to the scalar 
 curvature of the initial data, which we were able to bound in terms of the angular momentum.
 The natural question is whether one can write a different quasi-local quantity $E$, having both properties, that is, $E\to m_{ADM}$ and $E\sim H^2$ and such 
 that it produces a better, sharper 
 inequality.

 \vspace{0.5cm}

  The inequality \eqref{mainineq} is linear in the ADM mass and quadratic in the
  angular momentum as a result of the linear dependence of $dE_G/dt$ with the
  scalar curvature $\bar R$. This is to be confronted with the linear relation
  suggested by the Bekenstein conjecture for the entropy of macroscopic objects
  \cite{PhysRevD.23.287} \cite{PhysRevD.61.024018}. The positivity of the
  entropy function implies the bound $\mathcal E\geq |J|/\mathcal R$ where
  $\mathcal E$ is the total energy and $\mathcal R$ is the radius of the
  smallest sphere that encloses the object (note the similar dependence for the
  case of black holes  \cite{Reiris:2013jaa}). However, the quadratic relation
  seen in \eqref{mainineq} is in accordance with the Newtonian limit
  \eqref{Enewton2}. Indeed, the paralelism is clear if one takes the hoop radius
  as the areal radius $\RA$. An important question is whether this discrepancy between the Bekenstein 
  conjecture and our own result for the relation between mass and angular momentum is due to the hypotheses in 
  our theorem. We do not have a clear answer for that, but we note, as 
pointed out by Unruh and Wald \cite{Unruh:1982ic}, that the Bekenstein bound is not essential for the validity of the 
generalized laws of thermodynamics \footnote{We thank an anonimous referee for making this observation to us.}, and might not be optimal.

\vspace{0.5cm}
 
 Disregarding the non-negative $m_T$ term and re-writing inequality \eqref{mainineq} in the form
 \be
 \RC^2\geq\frac{1}{5}\frac{J^2}{m_{ADM}\RA}
 \ee
we see that if we fix the total mass and the area of the object, then the greater the angular momentum, the greater the circumferential radius. This result agrees with our expectation and experience 
that rotation produces flattening. We see that rotation sets restrictions onto how prolate an ordinary object can be. Also, this inequality gives information about
localization in space. It implies that for given total energy $m_{ADM}$, a rotating region can not be too small. Note that there is an important difference between 
this kind of argument and the ones used in black hole formation criteria, where only quasi-local quantities are taken into account. We will come back to 
this issue below.

\vspace{0.5cm}

 No equations of state were assumed. Inequality \eqref{mainineq} is a consequence of Einstein constraint equations and does not occur in pure Newtonian theory
 (recall that in the introduction we obtained a similar inequality 
\eqref{Enewton2} after assuming the Hoop condition \eqref{hoop}). Matter enters 
the inequality only 
 via the $m_T$ term and the dominant energy condition needed to make the Geroch energy have the positivity and monotonicity properties.

\vspace{0.5cm}

When Maxwell fields are taken into account, the inequality can be extended 
using similar techniques, provided that there is no electromagnetic contribution to the angular momentum outside the body. For the treatment of the electromagnetic contributions 
we  follow the work of Dain \cite{Dain:2015ira}, where the case of time 
symmetric data with no rotation, $J=0$, was considered.

We write the energy density in the form
\be
\mu=\mu_{\mbox{(not EM)}}+\frac{E^2+B^2}{8\pi}
\ee
where $\mu_{\mbox{(not EM)}}$ stands for non-electromagnetic matter fields 
satisfying the dominant energy condition, and $E^i$, $B^i$ are 
the electromagnetic fields. Then the integral in \eqref{evol2} has tree terms, 
the ones involving  $\mu_{\mbox{(not EM)}}$ and $K_{ij}  K^{ij}$ are 
treated in exactly the same manner as before. The term involving the 
electromagnetic contribution to the energy density could be bound in terms of the 
electric charge by following \cite{Jang:1979zz}.  We sketch the proof below.
\be
\begin{split}
 \frac{A_t^{1/2}}{(16\pi)^{1/2}} \int_{S_t}  \mu &= 
\frac{A_t^{1/2}}{(16\pi)^{3/2}} \int_{S_t} 2 (E^2 + B^2)\geq \frac{2 
A_t^{1/2}}{(16\pi)^{3/2}} \int_{S_t}  (E^j \nu_j)^2 \\
 &\geq \frac{2 A_t^{1/2}}{(16\pi)^{3/2}} \frac{\left(\int_{S_t}  E^j \nu_j 
\right)^2 }{A_t} \geq \frac{2 \left(\int_{S_t}  E^j \nu_j\right)^2 
}{(16\pi)^{3/2} A_t^{1/2} }   =
 \frac{\sqrt{\pi} Q_t^2}{2 A_t^{1/2}}. 
\end{split}
\ee
In the third step we used the H\"older inequality, and in the fifth step we used 
the Gauss theorem and the definition of electric charge
\be
Q_t=\frac{1}{8\pi}\int_{S_t}E^j \nu_j.
\ee
Let $T$ be the smallest value of the flow parameter such that for $t\geq T$
the surfaces given by the flow are convex and such that it encloses the object 
(we still need this convexity condition to control the rotation part 
of the evolution). 
Then since $S_t$ lays outside the object, the electric charge is 
$Q_t=Q_T:=Q$, and thus for $t>T$ the time derivative of the Geroch energy is 
bounded by
\be
\label{evol6}
\frac{d}{dt}E_G\geq    \frac{\sqrt{\pi}A_t^{1/2}}{\int_{S_t} \eta }J^2 + \frac{\sqrt{\pi} }{2 A_t^{1/2}} Q^2.
\ee

Integrating \eqref{evol2} from 0 to infinity, using \eqref{evol6} and
 \eqref{propEG}, we obtain:
\be\label{mainineqQ}
m_{ADM}\geq  m_{T} +\frac{1}{2} \frac{ Q^2}{ \RA} + 
\frac{1}{5}\frac{J^2}{ \RA \RC^2}. 
\ee
where again $m_T$, given by \eqref{mT}, comes from integrating \eqref{evol2} in 
the interval $t\in (0,T)$. Note that the electric part of 
this inequality does not depend on the circumferential radius, only the size 
(measured by the surface area) is relevant in that case.

\vspace{0.5cm}

The final remark we want to make is about where we start the IMCF. In
our theorem and treatment so far, we started out from a point on the symmetry axis in $M$ and cover
all $M$ with the surfaces $S_t$. In particular, we look at one of these
surfaces, that we call $S_T$, to obtain information about the object's physical
and geometrical properties. However that might not be the most convenient way of
studying the object, as one might lose control over where $S_T$ is or how far
away from the object it is. We know from Huisken and Ilmanen's work \cite{Huisken01} that for sufficienty large times, the surfaces are convex (as 
they approach round spheres at infinity), but clearly, one of such surfaces near infinity would not give a good description of the object's size. 
Unfortunately, one does not \textit{a priori} know where $S_T$ will be located.
An alternative approach is to start the flow from a convex
surface $S_0$, chosen in such a way that its evolution preserves convexity and,
more importantly, such that it coincides with the object's surface or it is the
smallest surface enclosing it.  By following this procedure we arrive at the
following inequality
\be\label{mainineq2}
m_{ADM}\geq E_G(S_0) +\frac{1}{5}\frac{J^2}{\RA \RC^2}
\ee
where now $E_G(S_0)$ is the Geroch energy of the initial surface $S_0$ and the quantities $J$, $\RA$, $\RC$ refer to $S_0$ as well.
This approach is particularly useful for numerical calculations and it is 
the one we use in the next section.

\vspace{1cm}
There are important open questions we want to address next. The first one being the convexity condition, we present some ideas and numerical results in the next section. 
Secondly is the appropriate notion of size one should use. In axial symmetry, the areal 
and circumferential radii are well defined and are related to relevant properties of the region under study, that is, localization and rotation. However, extending these ideas outside axial 
symmetry does not seem straightforward and a more general measure of size should be introduced. Another issue has to do with the boundary between ordinary objects (like the ones we study 
here) and black holes. More precisely, could we use this inequality to formulate a black hole formation criteria similar to the one proved by Khuri \cite{Khuri:2015zla}? We see that there
are differences between our work and that of Khuri, because we include a global quantity, the ADM mass. Therefore, our inequality does not lend itself directly to an  argument of the 
kind 'if the angular momentum is too localized, then a black hole will form'. In our main result, the localization of angular momentum, represented by the ratio 
$(angular\; momentum)(size)^{-1}$ is compared to the total energy.
Finally, we would like to understand better the relation between our 
approach to study inequalities for ordinary objects and the arguments involved in 
the derivation of the Penrose inequality \cite{Anglada:2016}.

\section{Numerical tests}\label{sec:num}

\renewcommand{\arraystretch}{1.3}

An important ingredient in our result is the convexity of surfaces $S_t$ along
the IMCF evolution. In this section we want to show that, in particular
examples, this property of the flow holds even when the surfaces $S_t$ are close
to the object.  Also, we want to evaluate the relative importance of the term
involving the angular momentum in our inequalities. For these purposes
we want to compute the IMCF starting with an initial surface which
is good to represent the physical and geometrical properties of the object, and
therefore we study the inequality in the form \eqref{mainineq2} to choose an
initial surface $S_0$ that is \textit{as close as possible} to $\partial
\Omega$.
  
We expect that within our approach, the convexity of the flow could be relaxed.
This is suggested by the convergence properties of the last integral in
\eqref{evol4}. This point will be studied in detail \cite{Anglada:2016}.

Every numerical example in this section is computed in two stages.
First, the elliptic problem for the conformal factor $\Psi$ is solved, giving an
initial data set.  This initial data set is completely determined by a compact
material object with maximum angular momentum $J$ compatible with the dominant
energy condition (see Appendix \ref{appendix_A}). Secondly, the IMCF equation is
used to compute the evolution of a convex initial surface $S_0$ that tightly
encloses the object. The preservation of the surface's convexity along the flow
and our main inequality are then checked.

\subsection{Computation of the conformal factor $\Psi$}

We restrict our numerical examples to initial data sets which are maximal,
asymptotically flat, and conformally flat. The exact set up and assumptions, and
the derivation of the equations involved are described in the Appendix \ref{appendix_A}. 

Let $(\rho, \varphi, z)$ be cylindrical coordinates on the conformal, flat
geometry, adapted to the axial symmetry of our problem, where we need to solve
the equation for $\Psi$. Because of the axial symmetry, no function depends on
$\varphi$.
The conformal factor $\Psi$, and thus the initial data set defined on the
initial slice $M= {\mathbb R}^3,$ is the solution of the semi-linear
elliptic problem \eqref{eq:47},\eqref{eq:57}, deduced in the Appendix \ref{appendix_A}.
In cylindrical coordinates, this problem is
\begin{equation}\label{problem_psi}\begin{split}
\Delta \Psi &=- \frac{2\pi~a~\rho}{\Psi^3}-\frac{ |\partial f|^2 \rho^2
}{4\Psi^{7}}, \\
\partial_\rho \Psi(\rho=0)&=0, \quad \lim_{r \to \infty}\Psi(\rho,z) = 1, \quad r = \sqrt{\rho^2 + z^2},
\end{split}
\end{equation}
where $\Delta$ is the flat Laplacian and the function $f(\rho,z)$ is, in turn, a
solution of the linear elliptic problem
\begin{equation}\label{problem_f}\begin{split}
\partial^2_\rho f+\partial^2_z f+\frac{3\partial_\rho f}{\rho} &= -8\pi a,\\
\partial_\rho f(\rho=0)=0, \quad &\lim_{r \to \infty}f(\rho,z) = 0, \quad r = \sqrt{\rho^2 + z^2}.
\end{split}
\end{equation}
The positive function $a(\rho,z)$ appearing in the source of both equations is a free
function that determines the matter content and the angular momentum content of
the initial data. In our examples we choose this function to have compact
support.

The angular momentum content, equation \eqref{eq:56}, in any region $\Omega$ of
the initial slice is given by
\begin{equation}\label{J_of_Omega}
J = -\int_\Omega a \rho^2 dv_0.
\end{equation}
Here and in what follows, $dv_0$ denotes the volume element on the flat,
conformal geometry. Once $\Psi$ is computed, by solving the problem
\eqref{problem_psi},\eqref{problem_f}, the rest of the physically relevant
quantities (see Appendix \ref{appendix_B}) can be computed. The area of an axisymmetric surface $\partial\Omega$
is given by equation \eqref{eq:58}, from which areal radius
$\RA$ is obtained. The circumferential radius
$\RC = \mathcal{C}/2\pi,$ is computed by finding the
greatest axisymmetric circle  $\mathcal{C}$ of $\partial\Omega.$ The average
baryonic mass density of the object is $\rho_b = M_b/V$ with the baryonic mass
and volume of the object given by
\begin{equation}
M_b = \int_{\Omega} \frac{\rho\,a}{\Psi^2} dv_0,\quad \mbox{and} \quad V = \int_{\Omega} \Psi^6
dv_0. \label{b_mass_and_volume}
\end{equation}
Here $\Omega$ is the support of the function $a$.

The ADM mass of the initial data can be computed as a volume integral on the
whole space (equation \eqref{m_adm}).

To solve the problem \eqref{problem_psi},\eqref{problem_f}, we proceed as follows.
First, in all our examples we choose, for simplicity, matter content (function $a$)
to satisfy reflection symmetry on the plane $z=0,$ i.e., $a(\rho,-z) =
a(\rho,z).$ This property effectively reduces the problem to half size; one
needs to solve only for $z\ge 0.$ Regularity of the solution at $z=0$ becomes
homogeneous Neuman boundary condition at $z=0.$ Second, we compactify the
problem by introducing new coordinates so that the whole quarter $\rho$--$z$ plane 
maps to a unit square. The new coordinates are
\begin{equation}\label{xy_coordinates}
x = \frac{\rho}{\rho+\rho_H}, \quad y = \frac{z}{z+z_H},
\end{equation}
where the parameters $\rho_H>0$ and $z_H>0$ can be freely chosen. The symmetry
axis, $\rho=0$ maps to $x=0,$  $\rho=\infty$ maps to $x=1$ and $\rho=\rho_H$ maps
to $x=1/2.$ Analogously, the plane $z=0$ maps to $y=0,$ $z=\infty$ maps to
$y=1$, and $z=z_H$ maps to $y=1/2.$ Summarizing, the compact elliptic problem we need to
solve, in the square $0\le x,y\le 1,$ is
\begin{equation}\label{problem_f_xy}
\frac{(1-x)^4}{\rho_H^2} f_{xx} + \frac{(1-x)^3(3-2x)}{x\rho_H^2} f_x +
\frac{(1-y)^4}{z_H^2} f_{yy} - 2\frac{(1-y)^3}{z_H^2} f_y = -8\pi a,
\end{equation}
with boundary conditions
\begin{equation}
f_x(0,y) = f_y(x,0) = 0, \quad f(1,y) = f(x,1) = 0,
\end{equation}
and
\begin{multline}\label{problem_Psi_xy}
\frac{(1-x)^4}{\rho_H^2} \Psi_{xx} + \frac{(1-x)^3(1-2x)}{x\rho_H^2} \Psi_x +
\frac{(1-y)^4}{z_H^2} \Psi_{yy} - 2\frac{(1-y)^3}{z_H^2} \Psi_y \\
= -2\pi\rho_H\frac{x\,a}{(1-x)\Psi^3} - \frac{\rho_H^2}{4} \frac{x^2|\partial
f|}{(1-x)^2\Psi^7},
\end{multline}
with boundary conditions
\begin{equation}
\Psi_x(0,y) = \Psi_y(x,0) = 0,\quad \Psi(1,y) = \Psi(x,1) = 1.
\end{equation}
To compute the solution to these problems we use finite differences. We
discretize $x$ and $y$ in uniform grids
\begin{equation}\label{discrete_coordinates}
x_i=hi, \quad y_j=hj, \quad i,j = -2, -1, 0, 1, 2, \dots, N-1, N,
\end{equation}
where the mesh size is $h=1/N,$ and we choose $N = 2^k,$ being $k$ a positive
integer. The index values from 0 to $N$ cover the unit square
including the boundaries, while the ghost values $-2, -1$ are used to impose
the homogeneous Neuman boundary conditions.

All functions of the problem become grid-functions. To discretize the equations
we use standard difference operators which are fourth order accurate all over
the domain, centered in the interior of the domain, and semilateral and lateral
close to and on the outer border of the domain.

Let us denote $v_{i,j} = v(x_i,y_j)$ the solution of any of the equations
\eqref{problem_f_xy} or \eqref{problem_Psi_xy}. The values $v_{N,j},~j=0,1,\dots
N$ and $v_{i,N},~i=0,i,\dots,N$ are fixed by the Dirichlet boundary condition at
infinity. The homogeneous Neuman boundary conditions at $x=0$ and $y=0$ are
imposed by setting the ghost points $i=-2,-1$ and $j=-2,-1$ in the way
\begin{equation}\label{ghost_values}\begin{split}
v_{-2,j} = v_{2,j}, \quad v_{-1,j} = v_{1,j}, \quad j=0, 1, \dots, N,\\
v_{i,-2} = v_{i,2}, \quad v_{i,-1} = v_{i,1}, \quad i=0, 1, \dots, N.
\end{split}
\end{equation}

The algebraic linear system of equations obtained for $f_{i,j}$ could be solved
by a direct method (like $LU$-decomposition plus Gaussian elimination). The
algebraic, non-linear system of equations obtained for $\Psi_{i,j}$, however,
has to be solved by an iterative method. For simplicity we decided to use
iterative methods to solve both equations. 

To accelerate convergence, we use a
multigrid algorithm \cite{briggs_mg}. In all cases we use under relaxed Jacobi
smothers on the finest and intermediate grids, and over relaxed Jacobi on the
coarsest grid. The number of grid levels one can use in these problems, having
in mind the span of the difference operators, is upper bounded by $k-3.$
The grid functions are passed from a fine grid to the next coarser grid by
simple restriction. The prolongation of a grid function from a coarse
grid to a finer grid is carried out by cubic Hermite interpolation.
  
The initial iteration has to be chosen carefully in these iterative schemes, so
that the overall method converges. In some cases it is
enough to choose $f_{i,j} = 0$ and $\Psi_{i,j} = 1,$ but sometimes, on fine
grids, it is necessary to start with an interpolated coarser-grid solution.

In our code the multigrid algorithm is applied as a sequence of
$V$-cycles\cite{briggs_mg}. After a number of $V$-cycles, the maximum norms of
the residual and the increment of the solution $\delta v_{i,j}$ on the finest
grid are checked.  If the relative values of these two norms, with respect of
the norm of the solution itself, are smaller than $\epsilon = 10^{-10},$ the
iterations are stopped. In the equation for $\Psi_{i,j}$, the non linearity is
treated with the full approximation storage (FAS) algorithm\cite{briggs_mg}. For
every example presented in this paper, we compute the conformal factor $\Psi$ on
three grids with $N=128$, $N=256$ and $N=512$, and use the three solutions to
check convergence of the difference scheme.
  
All the integrals used to compute physically relevant quantities are
approximated by the Simpson's rule.

\subsection{Numerical setup for the IMCF equation}

We want to solve numerically the ordinary problem given by
\eqref{IMCF_equation}-\eqref{Lu_operator}, starting
with an initial surface at $t=0$ that is convex and that just encloses the
object under study. This problem, though differentially ordinary, is tricky to
be solved numerically. This is so because its exact solutions diverge
exponentially with time. The problem is therefore unstable; one has to be
careful when choosing a numerical scheme to approximate its solution, since
any perturbation (deviation from the exact solution) will also grow exponentially
with time. What we need is a conditionally stable scheme, which basically means
that the numerically computed approximation does not diverge faster than the
exact solution\cite{Kreiss-Ortiz-book}.

We discretize the variable $\theta$ on a uniform grid and approximate the
derivatives with respect to $\theta$ by standard, centered, fourth order accurate finite
difference operators. The regularity of $v(\theta,t)$ at $\theta=0$ and the
regularity plus reflection symmetry in the $z=0$ plane imply Neuman boundary
conditions at $\theta=0$ and $\theta=\pi/2.$ These boundary conditions are well
handled (imposed) by using two ghost points on each side of the interval
$[0,\pi/2].$ 

To integrate in time, one can choose a fourth order accurate method too.
However, this is not worth the price. The reason is that the eigenvalues of the
linearized equation are positive and large and then, to have a conditionally
stable scheme, the time step has to be very small as compared with the mesh
size and a low order method can do the job. The explicit Euler method turned out
to be appropriate for our problem because the overall method (fourth order in
space and first order in time) becomes conditionally stable for the
discretizations we use when the time step is comparable to the fourth power of
the mesh size. If the time step is not small enough, regardless of the time
integration method, the method becomes unstable and the solution breaks with
peaks that diverge in a few time steps. In our calculations we use a mesh size
$\delta_\theta = \pi/200$ and time steps as small $2\times 10^{-6}.$

\subsection{Object models and results}

In our setup an initial data is completely defined by the function $a(\rho,z)$.
The finite difference approximations described in the previous section have
truncation errors that involve, in the leading term, fifth and sixth order
derivatives of the solution. Therefore, to not ruin the accuracy of the computed
approximation and the convergence rate of the iterative methods, the function
$a$ defining the object has to be at least $C^4$-smooth on the computational
domain. So, we study objects given by compactly supported functions $a(\rho,z)$
which are defined in terms of the cutoff polynomial $q(s)$ given by
\begin{equation}\label{q_polynomial}
q(s) = 1 - s^4\Bigl(1 - 5(s-1) + 15(s-1)^2 - 35(s-1)^3 + 70(s-1)^4\Bigr).
\end{equation}
All the magnitudes given in this section are in geometrical units in the cgs
system.

\subsubsection{Spheroidal objects}

For these examples we choose the function $a$ to have support on an axially
symmetric ellipsoid, which can be either oblate or prolate (being a sphere a
particular case),
\begin{equation}\label{alpha_spheroids}
a(\rho,z) = \begin{cases}
a_0\, q(s), & 0\le s< 1,\\
0, & 1\le s,
\end{cases}
\end{equation}
where
\[
\quad s =\sqrt{\Bigl(\frac{\rho}{R}\Bigr)^2 + \Bigl(\frac{z}{Z}\Bigr)^2}
\]
and $R, Z$ are positive constant. The functions $a(\rho,z)$ so defined are
$C^4$-smooth on the whole domain. The support region $\Omega$ is given by $s<1,$
while the surface of the object, $\partial\Omega$, is given by $s=1.$
 
In Table \ref{table_spheroidal_objects} we present various spheroidal objects of
different densities and sizes. The parameters $a_0$, $R$ and $Z$ are displayed
together with the resulting angular momentum $J$. The object called $NS$ is
chosen to be slightly oblate, and has parameters so that the baryonic mass
density and object size are comparable to those of a neutron star (see Table
\ref{table_object_surfaces} and, for example, \cite{0004-637X-757-1-55}). The following three objects, $P$, $O$ and
$VO$, have the same value of $a_0$ and similar size parameters, which result
in comparable baryonic densities. $P$ is a prolate spheroid; $O$ an
oblate spheroid and $VO$ a very oblate spheroid. The object called $S$ is a
larger and lighter, slightly prolate, spheroid with parameters chosen so that
the resulting size and baryonic mass density are comparable to those of the sun.

\begin{table}[t]
\begin{center}
{\small
\begin{tabular}{@{}lcccc@{}}\toprule
{\bf Obj.} & $a_0$ & $R$ & $Z$ & $J$ \\ 
\hline
NS & $6.00\times 10^{-18}$ & $7.50\times 10^5$ & $7.30\times 10^5$ & $1.4614\times 10^{11}$ \\ 
P & $6.00\times 10^{-18}$ & $0.80\times 10^6$ & $1.00\times 10^6$ & $2.5916\times 10^{11}$ \\ 
O & $6.00\times 10^{-18}$ & $1.00\times 10^6$ & $0.80\times 10^6$ & $5.0617\times 10^{11}$ \\ 
VO & $6.00\times 10^{-18}$ & $1.00\times 10^6$ & $0.30\times 10^6$ & $1.8981\times 10^{11}$ \\ 
S & $2.90\times 10^{-38}$ & $6.70\times 10^{10}$ & $6.5\times 10^{10}$ & $4.0056\times 10^{15}$ \\ 
\bottomrule
\end{tabular}
}
\caption{Parameters defining various spheroidal objects.}\label{table_spheroidal_objects}
\end{center}
\end{table}
In Table \ref{table_object_surfaces} we show, for each object in Table
\ref{table_spheroidal_objects}, the average baryonic density of the
object, the physical size of the object represented by the radii
$\RA$ and $\RC$ of the object's surface
and the ADM mass $m_{ADM}$ of the initial data. The last column in the table
indicates whether the object's surface turns out to be convex or not.
\begin{table}[t]
\begin{center}
{\small
\begin{tabular}{@{}lccccc@{}}\toprule
{\bf Obj.} & $\rho_b$ & $\RC$ & $\RA$ & $m_{ADM}$ & convex? \\ 
\hline
NS & $4.3247\times 10^{-14}$ & $1.0461\times 10^6$ & $1.0227\times 10^6$ & $2.7103\times 10^5$ & yes \\ 
P & $3.1137\times 10^{-14}$ & $1.2309\times 10^{6}$ & $1.2859\times 10^6$ & $3.9647\times 10^5$ & yes \\ 
O & $2.3503\times 10^{-14}$ & $1.6462\times 10^{6}$ & $1.5379\times 10^6$ & $5.6641\times 10^5$ & yes \\ 
VO & $5.7604\times 10^{-14}$ & $1.3199\times 10^6$ & $1.0826\times 10^6$ & $2.8427\times 10^5$ & no \\ 
S & $1.1212\times 10^{-28}$ & $6.7000\times 10^{10}$ & $6.5796\times 10^{10}$ & $1.3695\times 10^5$ & yes \\ 
\bottomrule
\end{tabular}
}
\caption{Physical quantities for the objects defined in Table
\ref{table_spheroidal_objects}.}\label{table_object_surfaces}
\end{center}
\end{table}

\begin{table}[b]
\begin{center}
{\small
\begin{tabular}{@{}lccccc@{}}\toprule
{\bf Obj.} & $R$ & $Z$ & $\RC$ & $\RA$ & $J^2/(5 \RA \RC^2)$ \\
\hline
NS & $7.50\times 10^5$ & $7.30\times 10^5$ & $1.0461\times 10^6$ & $1.0227\times 10^6$ & $3.8168\times 10^3$ \\
P & $0.80\times 10^6$ & $1.00\times 10^6$ & $1.2309\times 10^6$ & $1.2859\times 10^6$ &  $6.8942\times 10^3$\\
O & $1.00\times 10^6$ & $0.80\times 10^6$ & $1.6462\times 10^{6}$ & $1.5379\times 10^6$ & $1.2295\times 10^4$ \\
VO & $1.00\times 10^6$ & $0.30\times 10^6$ & $1.3199\times 10^6$ & $1.1497\times 10^6$ & $3.5984\times 10^3$ \\
S & $6.70\times 10^{10}$ & $6.50\times 10^{10}$ & $6.7000\times 10^{10}$ & $6.5796\times 10^{10}$ & $1.0865\times 10^{-2}$ \\
\bottomrule
\end{tabular}
}
\caption{Parameters for the initial surface $S_0$  that preserves convexity when
evolved with the IMCF. The geometrical size parameters $\RC$
and $\RA$ that measure the size of the object are given
together with the quotient on the right hand side in the inequality
\eqref{mainineq}. This last quantity should be compared with the $m_{ADM}$
given in the Table \ref{table_object_surfaces}.}\label{table_NT_surfaces}
\end{center}
\end{table}
In Table \ref{table_NT_surfaces} we show, for each object in Table
\ref{table_spheroidal_objects}, the parameters for the smallest spheroidal
surface that encloses the object and is convex. In the cases where the object's
surface $\partial\Omega$ is convex, we take it as the initial surface $S_0$. In
cases where the object's surface is not convex, we take a surface, also defined
as an ellipsoid, with larger $R$ or $Z$ so that the surface is convex. We
computed the evolution of these surfaces and found that in all cases the
convexity is preserved by the IMCF evolution. As examples of these evolutions we
show in Figure \ref{fig_imcf_surfaces} plots of the surfaces for various times
for the NS and VO cases.
\begin{figure}[ht]
\includegraphics[width=6cm]{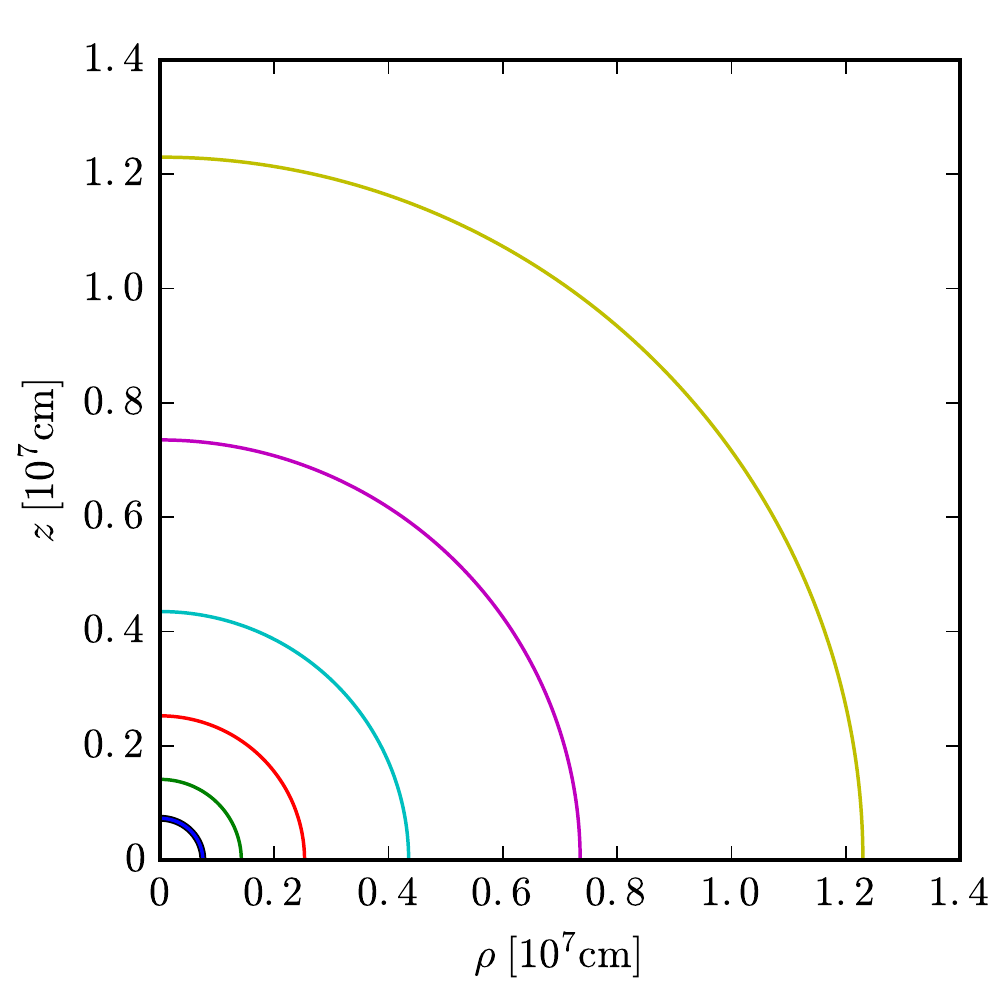}\includegraphics[width=6cm]{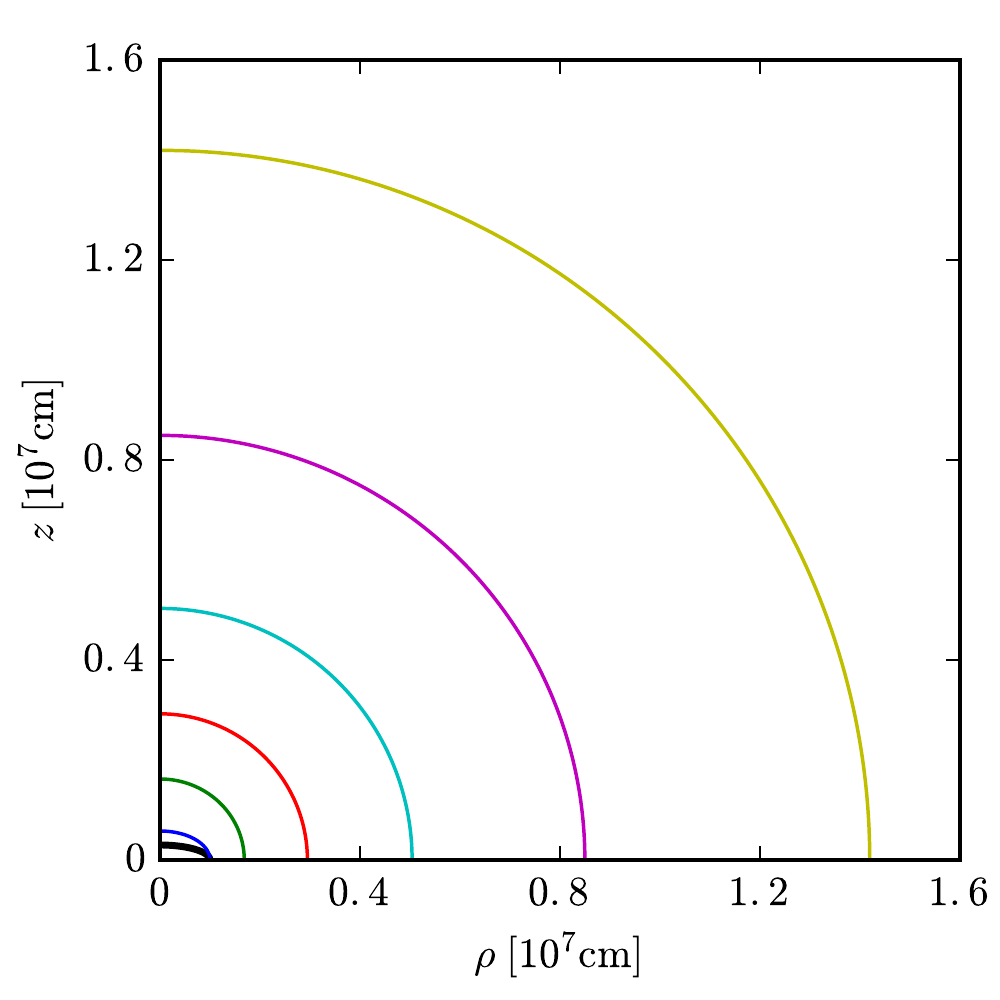}
\caption{IMCF evolution of the surfaces $S_t$ corresponding to the objects NS
(left plot) and VO (right plot) of Table \ref{table_NT_surfaces}. The plots
show, from the origin outwards, the surface of the object (thick curve), and the
surfaces at times $t=0, 1.0, 2.0, 3.0,  4.0, 5.0$ in a piece of the $\rho$--$z$ plane on the
flat geometry. In the NS case, the initial surface is coincident with the
object's initial surface. In the VO case the initial surface is larger than the object's
surface, so chosen so that it is convex.}\label{fig_imcf_surfaces}
\end{figure}

Figure \ref{fig_princ_curv} shows the plot of the principal curvatures (see
equations \eqref{lambda_phi} and \eqref{lambda_theta} in Appendic \ref{appendix_C}) of the
object's surface for the cases  NS and VO. In the first case the positivity of both
curvatures show the surface of the body is convex. In the second case the
principal curvatures become negative, showing the body surface is not convex. 
\begin{figure}[ht]
\begin{center}
\includegraphics[width=6cm]{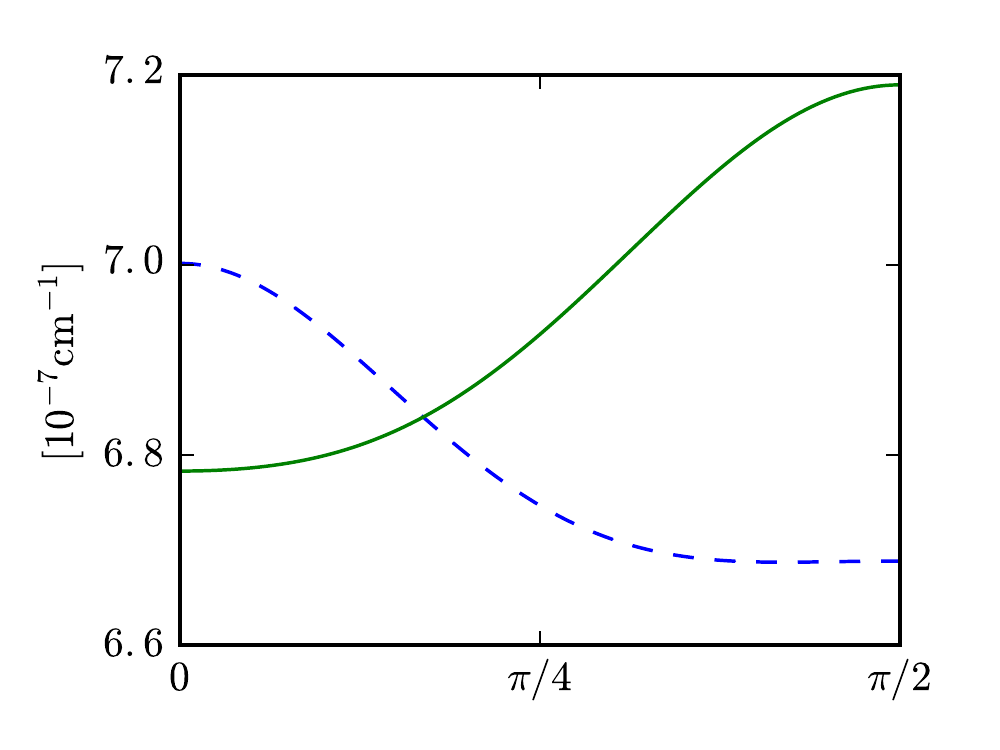}\includegraphics[width=6cm]{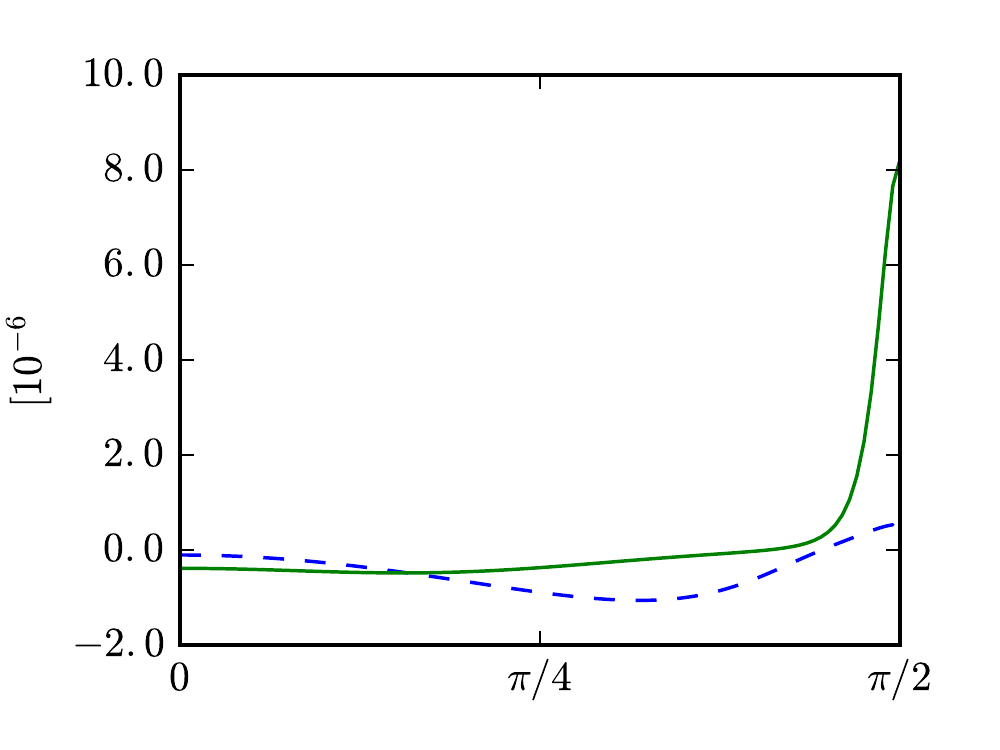}
\end{center}
\caption{Principal curvatures of the object's surface for the cases NS (left
plot) and VO (right plot) of Table \ref{table_spheroidal_objects}.
$\lambda_{\varphi}$ is represented in dashed lines and $\lambda_{\theta}$ in solid lines.}\label{fig_princ_curv}
\end{figure}

It is important to check, in all the examples, that the surfaces $S_t$ not only
remains convex along the evolution, but also that they approach spheres as time
increases.
This is clearly seen in plots of the principal curvatures $\lambda_\theta$, $\lambda_\varphi$. In the Figures
\ref{fig_P_ul} and \ref{fig_VO_ul} we plot the quotients
$\lambda_\theta(\theta)/(1/r(\theta))$ and
$\lambda_\varphi(\theta)/(1/r(\theta))$, as functions of $\theta,$ for various times, for the two more
extreme cases in Table \ref{table_object_surfaces}: P and VO objects. Both
figures show how the principal curvatures of $S_t$ approach the principal
curvatures of a sphere as $t$ grows.
\begin{figure}[h]
\begin{center}
\includegraphics[width=12cm]{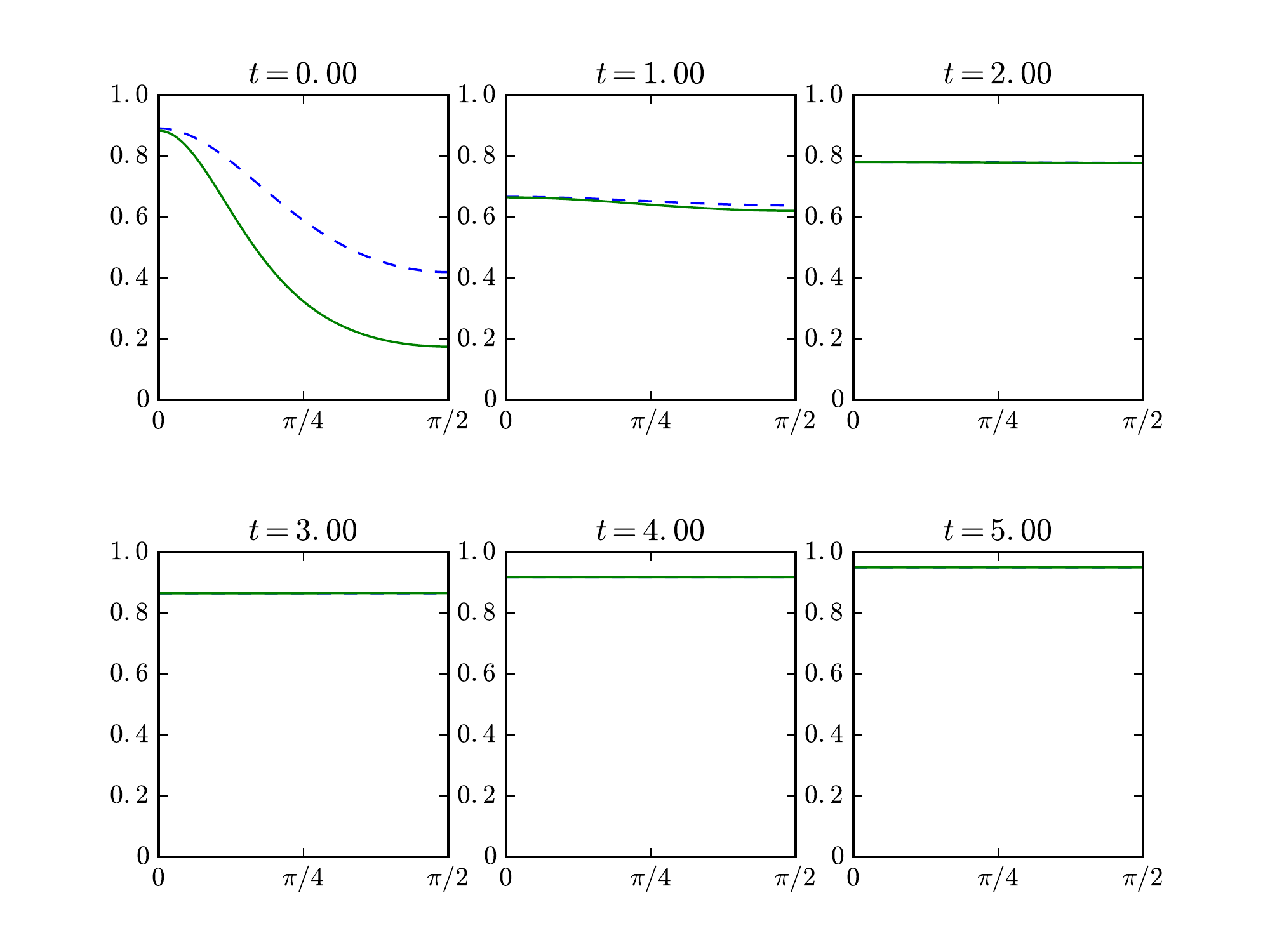}
\caption{Plots of $r\lambda_\varphi$ (dashed lines) and $r \lambda_\theta$
(solid lines) as functions of $\theta$ along the IMCF evolution of the surface
of Table \ref{table_NT_surfaces} for the object P.}\label{fig_P_ul}
\end{center}
\end{figure}
\begin{figure}[h]
\begin{center}
\includegraphics[width=12cm]{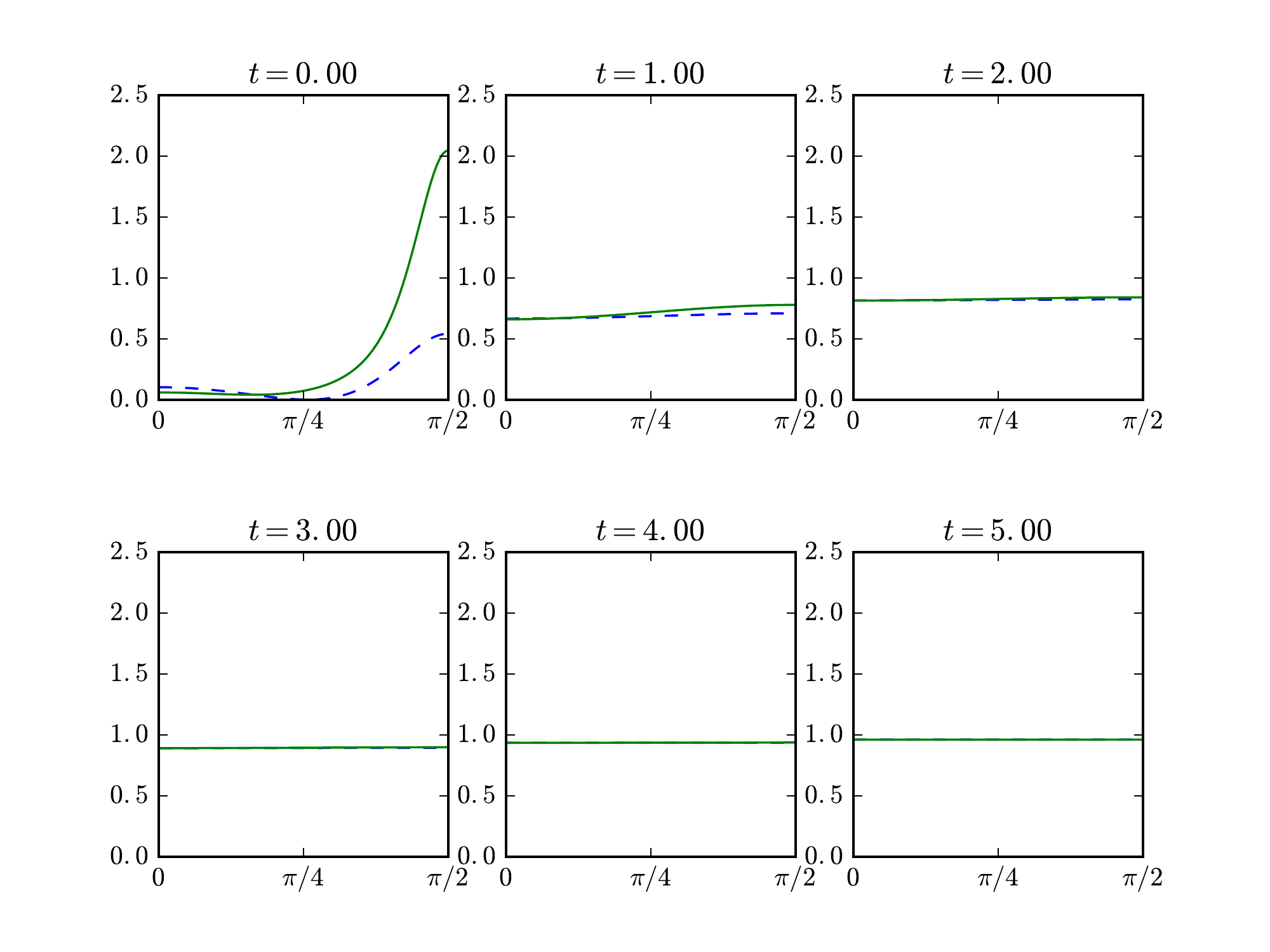}
\caption{Plots of $r\lambda_\varphi$ (dashed lines) and $r \lambda_\theta$
(solid lines) as functions of $\theta$ along the IMCF evolution of the surface
of Table \ref{table_NT_surfaces} for the object VO.}\label{fig_VO_ul}
\end{center}
\end{figure}

\subsubsection{Concave object}\label{sec_532}

As a final example we compute the initial data corresponding to an object which
is concave even as seen on the conformal flat geometry. We study an object
whose function $a(\rho,z)$ is given by equation \eqref{alpha_spheroids} but in
this case the parameter $s$ is defined as
\begin{equation}\label{s_concave}
s = \frac{\sqrt{\rho^2 + z^2}}{B\Bigl(D + \frac{\rho^2}{\rho^2+z^2}\Bigr)}.
\end{equation}
As initial surface $S_0$ to evolve the IMCF we use a spheroidal surface as
before. The object parameters we use in this example are: $B = 5.0\times 10^5$,
$D = 0.3$, $a_0 = 6.0\times 10^{-18}.$ Thus $J = 4.1743\times 10^{10}$, and we
obtain $m_{ADM} = 1.0982\times 10^5.$ The parameters we use for the convex
spheroidal surface $S_0$ enclosing the body are: $R = 6.60\times 10^5,$ $Z =
3.8\times 10^5$ and we obtain $\RA = 6.7350\times 10^5$ and $\RC = 7.7969\times
10^5$. We verify that the IMCF evolution of $S_0$ preserves convexity and
approaches spheres as in the previous cases. A plot of the object surface,
initial surface and short time evolution is shown in Figure \ref{fig_C2}.
\begin{figure}[h]
\begin{center}
\includegraphics[width=6cm]{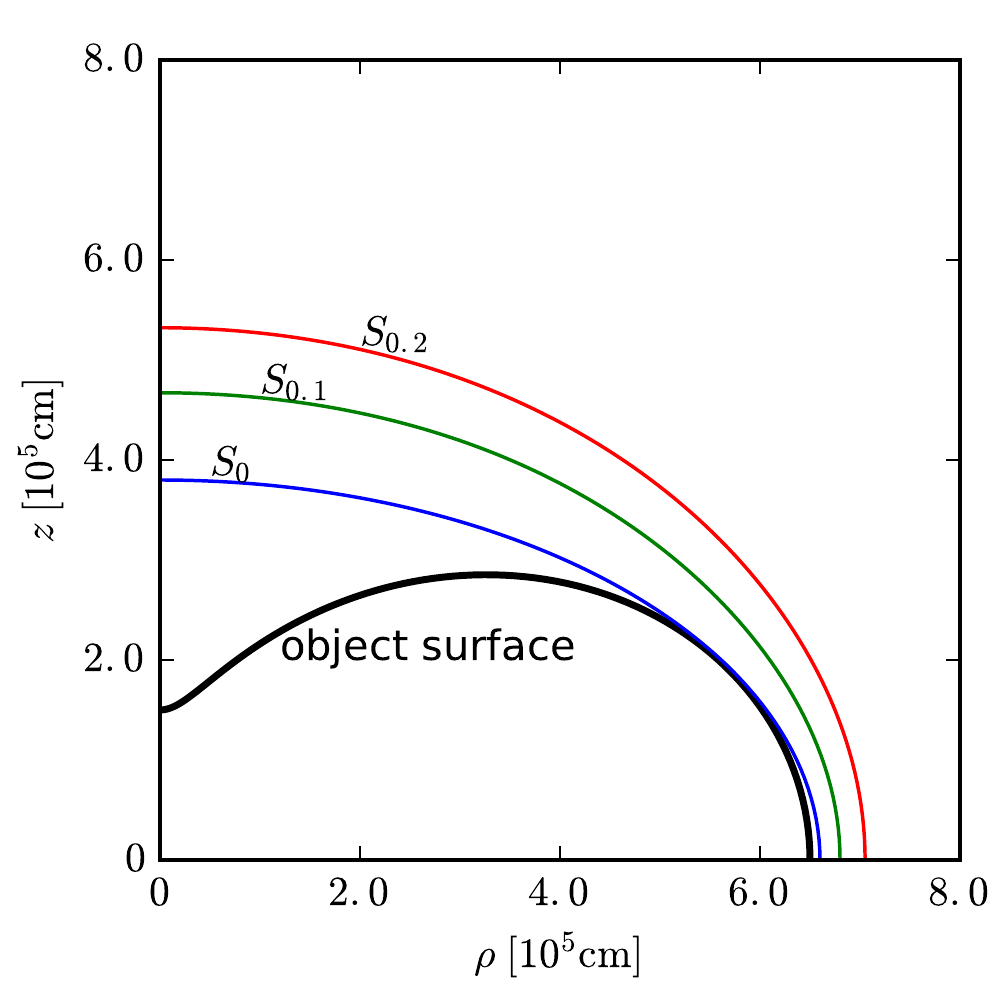}
\caption{Concave object's surface and short time IMCF evolution of the convex spheroidal
surface $S_0$ enclosing it.}\label{fig_C2}
\end{center}
\end{figure}

\subsection{Discussion of numerical results}

Our purpose with the numerical examples computed in this section is twofold:
first, to check the convexity condition of the surfaces $S_t$ along the IMCF
evolution, and second to study the relevance of the
$J^2/(5\RA \RC^2)$ term on the right hand
side in the inequality \eqref{mainineq2} in various cases. We want
to gain insight about when the angular momentum term on the right hand side of the
inequality becomes an important contribution as compared with the ADM mass.

The convexity condition of the surfaces along the IMCF evolution is 
clearly verified in all examples we computed. When the object's surface is convex, one
can evolve the flow starting with $S_0 = \partial\Omega$ (the surface of the
object) and the convexity is preserved by the evolution. When the object's surface
is not convex it is enough to choose, it seems, a spheroidal initial surface
$S_0$ that just encloses the object and the evolution preserves the convexity
again. 

The shape of the objects we consider are simple and our numerical examples far
from exhaustive, however we do not find any sign that the IMCF evolution will
violate the convexity, at least under the smoothness conditions we impose in our
numerical examples. This point certainly deserves a deeper study that we intend
to carry out in future works.

To gain insight on the relative importance of the term proportional to $J^2$ on
the right hand side of \eqref{mainineq2} we define the ratio
\be
\Gamma:= \frac{\Bigl(\frac{J^2}{5\RA \RC^2}\Bigr)}{m_{ADM}}. 
\ee

The values of $\Gamma$ for the high density  objects in Table \ref{table_spheroidal_objects}
are
\[
\Gamma_{NS} \simeq 1.4\times 10^{-2}, \quad \Gamma_{P} \simeq 1.7\times 10^{-2},
\quad \Gamma_{O} \simeq 2.2\times 10^{-2}, \quad \Gamma_{VO} \simeq 1.3\times
10^{-2},
\]
while for the lower density object S we have
\[
\Gamma_{S} \simeq 7.9\times 10^{-8}.
\]
For the high density concave object of Section \ref{sec_532} we get $\Gamma
\simeq 7.6\times 10^{-3}.$

Though we have not computed many examples, our results suggest that the contribution
of the $J^2$-term is higher for higher density objects. By comparing the values
of $\Gamma$ for the objects $P$ and $O$, with interchanged values for the
parameters $R$ and $Z$, we see that, as expected, the value of $\Gamma$ is
higher for the oblate spheroid than for the prolate spheroid.

\section*{Acknowledgments}
P. Anglada and O. E. Ortiz started studying the relationship between size and
angular momentum for axially symetric objects in collaboration with Sergio Dain
more than two years ago. All the authors feel deeply indebted to Sergio.
M. E. Gabach-Clement wants to thank Marc Mars, Markus Khuri and Greg Galloway
for illuminating discussions. This work was supported by grants from CONICET and
SECyT, UNC.

\setcounter{section}{0}
\setcounter{subsection}{0}
\renewcommand\thesection{\Alph{section}}
\renewcommand\thesubsection{\thesection.\arabic{subsection}}

\appendix
\section{Conformal method}\label{appendix_A}

In this section we derive the constraint equations in the 
conformally flat case, which is one of the set of equations (together with the flow equations, given in appendix \ref{appendix_C}) 
we use we use in our numerical computations of section \ref{sec:num}.

The conformal method is a well known technique that can be used to simplify the constraints 
\begin{align}
 \label{const1}
   \bar D_\beta   K^{\alpha \beta} -  \bar D^\alpha   (\mbox{tr}K)= -8\pi j^\alpha,\\
 \label{const2}
   \bar R -  K_{\alpha \beta}   K^{\alpha \beta}+  (\mbox{tr}K)^2=16\pi \mu,
\end{align}
where ${\bar D}$ and $\bar R$ are the Levi-Civita connection and
the curvature scalar associated with $  \bar g$. 

We restrict to the maximal case (tr$K$=0) and take $\tilde g_{\alpha \beta}$, and  $\tilde K_{\alpha \beta}$ to be symmetric tensor fields such that
\begin{equation}
  \label{eq:7}
   \bar g_{\alpha \beta}=\Psi^4 \tilde g_{\alpha \beta} \qquad K_{\alpha \beta}= \Psi^{-2} \tilde K_{\alpha \beta}
\end{equation}
where $\Psi$ is the positive conformal factor.

In terms of these new fields the constraint equations read
 \begin{align}
 \label{const1c}
   \tilde D_\beta   \tilde K^{\alpha \beta} = -8\pi \tilde j^\alpha ,\\
 \label{const2c}
   L_{\tilde g} \Psi= - \frac{2\pi \tilde \mu}{\Psi^3}-   \frac{\tilde K_{\alpha 
\beta}  \tilde
     K^{\alpha \beta}}{8\Psi^7}, 
\end{align} 
where $\tilde D$ is the covariant derivative with respect to $\tilde g$, and we have defined
\begin{equation}
  \label{eq:37}
  L_{\tilde g} \Psi= \Delta_{\tilde g} \Psi-\frac{1}{8} \Psi \tilde R,    
\end{equation} 
\begin{equation}
  \label{eq:12}
  \tilde \mu =\Psi^8\mu, \quad \tilde j^\alpha = \Psi^{10}j^\alpha 
\end{equation}
where $\tilde R$ is the scalar curvature associated to $\tilde g$.

The momentum constraint (\ref{const1c}) can be solved in the following form. Let $A_\alpha$ be a 1-form to be specified and set
\begin{equation}
  \label{eq:39}
  \tilde K_{\alpha \beta}= (\mathcal{L}A)_{\alpha \beta}+\sigma_{\alpha \beta},
\end{equation}
where
\begin{equation}
  \label{eq:40}
  (\mathcal{L}A)_{\alpha \beta}=\tilde D_\alpha  A_\beta +\tilde D_\beta A_\alpha -\tilde g_{\alpha \beta} \frac{2}{3}\tilde D_\gamma A^\gamma, 
\end{equation}
and $\sigma_{\alpha \beta}$ is an arbitrary trace-free tensor field. Then the momentum constraint
has the form
\begin{equation}
  \label{eq:41}
 \mathbf{L}_{\tilde g} A_\alpha  =-8\pi\tilde j_\alpha  -\tilde D^\beta \sigma_{\alpha \beta} ,
\end{equation}
where $\mathbf L_{\tilde g}$ is the elliptic operator
\begin{equation}
  \label{eq:42}
    \mathbf{L}_{\tilde g} A_\alpha = \tilde D^\beta (\mathcal{L}A)_{\alpha \beta}. 
\end{equation}

It is important to note  that the set of elliptic equations
(\ref{eq:42}) and (\ref{const2c}) can be solved on a bounded domain prescribing
Dirichlet conditions for $\Psi$ and the analog to Neumann condition for $A_\alpha$, namely
\begin{equation}
  \label{eq:43}
 \nu^\alpha (\mathcal{L}A)_{\alpha \beta}=\phi_\beta 
\end{equation}
where $\nu^\alpha$ is the outer unit normal to the boundary and   $\phi_\beta$ is a function on the boundary.

\subsection{Conformally flat initial data}

Consider a conformally flat initial data, that is
\begin{equation}
  \label{eq:13}
  \tilde g_{\alpha \beta}=\delta_{\alpha \beta}
\end{equation}
where $\delta_{\alpha\beta}$ is the flat 3-metric. In this case the constraint equations simplify to
 \begin{align}
 \label{const1cf}
   \partial_\beta   \tilde K^{\alpha \beta} = -8\pi \tilde j^\alpha ,\\
 \label{const2cf}
 \Delta \Psi =- \frac{2\pi \tilde \mu}{\Psi^3}-\frac{ \tilde K_{\alpha \beta} \tilde
   K^{\alpha \beta}}{8\Psi^{7}}. 
\end{align} 
where $\partial_\alpha$ are partial derivatives, $\Delta$ is the flat Laplace operator in 3 dimensions and  indices are moved with 
respect to $\delta$.

The solution of equation \eqref{const1cf} is constructed as in the previous
section. We choose $\sigma_{\alpha \beta}=0$, hence  
\begin{equation}
  \label{eq:19}
 \tilde  K_{\alpha \beta}=\partial_\alpha  A_\beta+\partial_\beta A_\alpha 
 -\frac{2}{3}\delta_{\alpha \beta} \partial_\gamma A^\gamma,
\end{equation}
and equation  \eqref{const1cf} translates into
\begin{equation}
  \label{eq:20}
  \Delta A^\alpha +\frac{1}{3} \partial^\alpha   \partial_\beta A^\beta= -8\pi \tilde j^\alpha .
\end{equation}

Imposing the following condition on the sources
\begin{equation}
  \label{eq:18}
  \partial_\alpha  \tilde j^\alpha =0.
\end{equation} we obtain
\begin{equation}
  \label{eq:21}
 \partial_\alpha A^\alpha =0,
\end{equation}
and hence the final set of equation equivalent to \eqref{const1cf} are
\begin{equation}
  \label{eq:22}
    \Delta A^\alpha = -8\pi \tilde j^\alpha.
\end{equation}
The solution of this equation is given by the Green function
\begin{equation}
  \label{eq:23}
  A^\alpha (x)= 2 \int \frac{\tilde j^\alpha (x')}{|x-x'|} d^3x'. 
\end{equation}
where $|\quad |$ is the flat norm.

We prescribe $\tilde \mu$ to be exactly the border case in the
dominant energy condition, namely
\begin{equation}
  \label{eq:14}
  \tilde \mu=\sqrt{\tilde j^\alpha  \tilde j_\alpha }. 
\end{equation}
and therefore the only free data is $\tilde j_\alpha $. We choose this vector to be
\begin{equation}
  \label{eq:16}
  \tilde j^\alpha  =a \tilde\eta^\alpha  ,
\end{equation}
where $a$ is a smooth function of the coordinates and $\tilde\eta_\alpha $ is the flat Killing vector.
In spherical coordinates $(r,\theta,\varphi)$, the Killing vector is $  \tilde\eta^\alpha =\frac{\partial }{\partial \varphi}$ and 
$a=a(r,\theta)$. Then using (\ref{eq:16})  we find that the solution to (\ref{eq:22}) is given by
\begin{equation}
  \label{eq:45}
  A^\alpha=f(r, \theta)\tilde\eta^\alpha ,
\end{equation}
where $f$ satisfies 
\begin{equation}
  \label{eq:47}
  \partial^2_r f +4\frac{\partial_r f}{r} +
\frac{\partial^2_\theta f}{r^2}+\frac{3\cos\theta \partial_\theta f
}{r^2\sin\theta}=-8\pi a.
\end{equation}
with the boundary condition
\begin{equation}
  \label{eq:48}
\lim_{r\to \infty}  f = 0.
\end{equation}

It is straightforward to calculate $\tilde K_{\alpha \beta}$ from the expression
(\ref{eq:45}), using the Killing equation and the fact that
\begin{equation}
  \label{eq:54}
  \tilde\eta^\alpha  \partial_\alpha  f=0, 
\end{equation}
and we obtain the remarkable simple formulae
\begin{equation}
  \label{eq:64}
  \tilde K_{\alpha \beta}=2 \tilde\eta_{(\alpha }\partial_{\beta)}f \quad \quad 
\tilde K_{\alpha \beta}  \tilde K^{\alpha \beta}=2 |\partial f|^2 
r^2 \sin^2\theta. 
\end{equation}

Summarizing, the systems we consider for the numerical computations  are prescribed by giving an arbitrary, axially symmetric,
function $a$ of compact support. This function describes the location of
the matter sources. Given  $a$, the equations \eqref{eq:47} with the boundary condition \eqref{eq:48} and
\begin{equation}
  \label{eq:57}
  \Delta \Psi =- \frac{2\pi a r \sin \theta}{\Psi^3}-\frac{ |\partial f|^2 r^2 \sin^2 \theta }{4\Psi^{7}}, \quad
  \lim_{r\to \infty}  \Psi = 1.   
\end{equation}

Using the sub and super solution method, it can be proven that given a smooth $a$ there exists a unique solution for the non-linear elliptic 
equation \eqref{eq:57}.\\

\section{Surfaces on conformally flat data}\label{appendix_B}

All relevant physical parameters of the initial data and of surfaces on $M$ can be computed in terms of $a$ and $\Psi$. Let $\Omega$ be any domain that contains 
the matter fields, and let $\dvf$, $ds_0$ be the flat volume and surface element 
respectively.

The ADM mass is given by
 \begin{equation}
   \label{m_adm}
   m_{ADM} =\int_{\Rt} \left( \frac{|\partial \Psi|^2}{2\pi\Psi^2} 
 +   \frac{a r \sin \theta}{\Psi^4}  +  \frac{2 
|\partial f|^2 r^2 \sin^2 \theta}{16\pi\Psi^8} \right)\dvf  
 \end{equation}
The angular momentum is
\begin{align}
  \label{eq:56}
  J &=-\int_\Omega a r^2 \sin^2 \theta \dvf.
\end{align}
The baryonic mass (as a measure of the quasi-local mass) is given by:
\begin{equation}
 M_b =  \int_{\Omega}  \mu  \dv =  \int_{\Omega} \frac{a r \sin \theta}{\Psi^2}  \dvf .
\end{equation}
Concerning the measures of the size of $\Omega$ we calculate the surface area $A(\partial \Omega)$ 
\begin{equation}
  \label{eq:58}
  A(\partial\Omega)=\int_{\partial \Omega} \Psi^4 ds_0
\end{equation}
and the length of the greatest axisymmetric circle $\mathcal C(\partial \Omega)$ 
\begin{equation}
\mathcal C= 2\pi \max_{\partial \Omega} 
\left( \sqrt{\eta} \right)=2\pi \max_{\partial \Omega} \left( r \sin\theta
\Psi^2 \right).
\end{equation}

\section{The IMCF on conformally flat initial data}\label{appendix_C}

In this section we derive the evolution equation for the flow surfaces, used in our numerical examples of section \ref{sec:num}.

We study the IMCF in $M=\mathbb R^3$ with metric $\bar g_{\alpha\beta}=\psi^4 \delta_{\alpha \beta}$. It is convenient to 
introduce a level-set formulation of \eqref{eqIMCF}, where the evolving surfaces are given as level-sets of a scalar function $u$ via $$ S_t= \{x^\alpha : u(x^\alpha,t)=0\}$$
and \eqref{eqIMCF} is replaced by the degenerate elliptic equation:

\begin{equation}
 \label{imcfu}
 \bar D_\alpha \left( \frac{\bar g^{\alpha \beta} \pd_\beta u }{ |\pd u|_{\bar g}} \right) = |\pd u|_{\bar g}
\end{equation}

Note that in our model we can write the mean curvature $H$ of $S_t$ in terms of $\Psi$:

\begin{equation}
\label{HN}
 H=\bar D_\alpha\left( \frac{\bar g^{\alpha \beta} \pd_\beta u }{ |\pd u|_{\bar g}} \right)=\frac{1}{\Psi^2 |\partial u|}\left( \Delta u + 4\frac{\partial\Psi \cdot \partial u}{\Psi} - \frac{ \partial |\partial u|  \cdot \partial u  }{|\partial u|} \right)
\end{equation}
where all the derivatives and dot products are computed  with respect to the flat metric.

At this point it is convenient to define an operator $\mathcal P$ acting on $u$:
\begin{equation}
\label{defLu}
 \mathcal P u=\Delta u + 4\frac{\partial \Psi \cdot \partial u}{\Psi} - \frac{ \partial |\partial u|  \cdot \partial u  }{|\partial u|}
\end{equation}
thus, 
equation \eqref{imcfu} can be written as
\begin{equation}
\label{hl}
 |\partial u|=\frac{\mathcal P u}{\Psi^2 H}.
\end{equation}

Now,  instead of solving this elliptic equation,
we will solve an evolution equation for $u$ obtained by taking a total time
derivative of $u$ on the surface $S_t$:
\begin{align*}
 0&=\frac{du}{dt}= \frac{\pd u}{\pd x^\alpha} \frac{\pd x^\alpha}{\pd t } + \frac{\pd u}{\pd t}= \partial_\alpha u \frac{\nu^\alpha}{H} + \frac{\pd u}{\pd t}= \partial_{\alpha} u \frac{\partial^\alpha u}{\Psi^2 |\partial u|} \frac{1}{H}  + \frac{\pd u}{\pd t} \\
  &= \frac{|\partial u|}{\Psi^2 H} +\frac{\pd u}{\pd t}= \frac{|\partial u|^2}{\mathcal P u} +\frac{\pd u}{\pd t}
\end{align*}
where in the second step we have used the IMCF equation, and in the last one we use that $H$ is given by \eqref{hl}, thus we have:
\begin{equation}
\label{imcfl}
 \frac{\pd u}{\pd t}=-\frac{|\pd u|^2}{\mathcal P u}.
\end{equation}
We take the surface  $S_t$ to be given by
\begin{equation}
u= u(r, \theta, t) =r - v(\theta,t)=0
\end{equation}
and putting this into \eqref{imcfl}, we arrive at  the following evolution equation for $v$
\begin{equation}\label{IMCF_equation}
\frac{\pd v}{\pd t}=\frac{1+\left(\frac{\pd_\theta v}{v}\right)^2}{\tilde {\mathcal P}v}
\end{equation}
where we define the operator $\tilde{\mathcal P}$ acting on $v$ as
\begin{equation}\label{Lu_operator}
 \tilde{\mathcal P} v=\left(
 \frac{2}{v}-\frac{\pd_\theta (\sin\theta \pd_\theta v)}{v^2 \sin\theta}
 +\frac{4}{\Psi}\left(\pd_r \Psi \rvert_{r=v} - \frac{\pd_\theta v \pd_\theta \Psi}{v^2} \right)
 +\frac{\left(\frac{\pd_\theta v}{v^2}\right)^2 \left(v + \pd^{2}_\theta v \right)  }{1+\left(\frac{\pd_\theta v}{v}\right)^2} 
 \right).
\end{equation}

Finally, the principal curvatures of $S_t$ are 
\begin{equation}\label{lambda_phi}
\lambda_{\varphi}=\frac{2r^2 \pd_r \Psi \sin\theta+\Psi 
r\sin\theta-2\pd_\theta \Psi \sin\theta u'-\Psi\cos\theta u'}{\Psi^3 
r\sin\theta\sqrt{r^2+u'^2}}
\end{equation}
\begin{equation}\label{lambda_theta}
\lambda_{\theta}=\frac{r\left[
(r^2+u'^2)(r^2 \pd_r\Psi+r-u''\Psi)+u''u'^2\Psi\right]}{\Psi^3(r^2+u'^2)^{5/2}}
\end{equation}
and $S_t$ is convex if $\lambda_\varphi,\lambda_\theta\geq0$.
Note that in the limit of $\Psi\to1$ and $u=const.$ we obtain the principal 
curvatures of the round sphere.
\be
\lambda^{sphere}_\theta=\lambda^{sphere}_\varphi=\frac{1}{r}
\ee


\begin{thebibliography}{10}

\bibitem{Acena:2010ws}
Andres Acena, Sergio Dain, and Maria~E. Gabach~Clement.
\newblock {Horizon area--angular momentum inequality for a class of axially
  symmetric black holes}.
\newblock {\em Class.Quant.Grav.}, 28:105014, 2011.

\bibitem{Anglada:2016}
Pablo Anglada.
\newblock {\em In preparation}, 2017.

\bibitem{Anglada:2015tan}
Pablo Anglada, Sergio Dain, and Omar~E. Ortiz.
\newblock {Inequality between size and charge in spherical symmetry}.
\newblock {\em Phys. Rev.}, D93(4):044055, 2016.

\bibitem{Arnowitt62}
R.~Arnowitt, S.~Deser, and C.~W. Misner.
\newblock The dynamics of general relativity.
\newblock In L.~Witten, editor, {\em Gravitation: An Introduction to Current
  Research}, pages 227--265. Wiley, New York, 1962.

\bibitem{PhysRevD.23.287}
Jacob~D. Bekenstein.
\newblock Universal upper bound on the entropy-to-energy ratio for bounded
  systems.
\newblock {\em Phys. Rev. D}, 23:287--298, Jan 1981.

\bibitem{briggs_mg}
W.~L. Briggs, V.~E. Henson, and S.~F. McCormick.
\newblock {\em A Multigrid Tutorial}.
\newblock SIAM, Philadelphia, PA, 2nd edition, 2000.

\bibitem{Clement:2012vb}
Maria E.~Gabach Clement, Jose~Luis Jaramillo, and Martin Reiris.
\newblock {Proof of the area-angular momentum-charge inequality for
  axisymmetric black holes}.
\newblock {\em Class. Quant. Grav.}, 30:065017, 2013.

\bibitem{Clement:2015fqa}
Maria Eugenia~Gabach Clement, Martin Reiris, and Walter Simon.
\newblock {The area-angular momentum inequality for black holes in cosmological
  spacetimes}.
\newblock {\em Class. Quant. Grav.}, 32(14):145006, 2015.

\bibitem{Dain06c}
Sergio Dain.
\newblock Proof of the angular momentum-mass inequality for axisymmetric black
  holes.
\newblock {\em J. Differential Geometry}, 79(1):33--67, 2008.

\bibitem{Dain:2013gma}
Sergio Dain.
\newblock {Inequality between size and angular momentum for bodies}.
\newblock {\em Phys. Rev. Lett.}, 112:041101, 2014.

\bibitem{Dain:2015ira}
Sergio Dain.
\newblock {Bekenstein bounds and inequalities between size, charge, angular
  momentum and energy for bodies}.
\newblock {\em Phys. Rev.}, D92(4):044033, 2015.

\bibitem{Dain:2011pi}
Sergio Dain and Martin Reiris.
\newblock {Area - Angular momentum inequality for axisymmetric black holes}.
\newblock {\em Phys. Rev. Lett.}, 107:051101, 2011.

\bibitem{Reiris:2013jaa}
Maria~Eugenia Gabach-Clement and Martin Reiris.
\newblock {Shape of rotating black holes}.
\newblock {\em Phys. Rev.}, D88(4):044031, 2013.

\bibitem{Galloway:2008gc}
Gregory~J. Galloway and Niall O'Murchadha.
\newblock {Some remarks on the size of bodies and black holes}.
\newblock {\em Class. Quant. Grav.}, 25:105009, 2008.

\bibitem{haensel2007neutron}
P.~Haensel, A.Y. Potekhin, and D.G. Yakovlev.
\newblock {\em Neutron Stars 1: Equation of State and Structure}.
\newblock Astrophysics and Space Science Library. Springer New York, 2007.

\bibitem{Hartle:1967he}
James~B. Hartle.
\newblock {Slowly rotating relativistic stars. 1. Equations of structure}.
\newblock {\em Astrophys. J.}, 150:1005--1029, 1967.

\bibitem{Hennig:2008zy}
J\"org Hennig, Carla Cederbaum, and Marcus Ansorg.
\newblock {A universal inequality for axisymmetric and stationary black holes
  with surrounding matter in the Einstein-Maxwell theory}.
\newblock {\em Commun. Math. Phys.}, 293:449--467, 2010.

\bibitem{PhysRevD.61.024018}
Shahar Hod.
\newblock Universal entropy bound for rotating systems.
\newblock {\em Phys. Rev. D}, 61:024018, Dec 1999.

\bibitem{Huisken01}
G.~Huisken and T.~Ilmanen.
\newblock The inverse mean curvature flow and the {R}iemannian {P}enrose
  inequality.
\newblock {\em J. Differential Geometry}, 59:352--437, 2001.

\bibitem{Huiskenevol}
Gerhard Huisken.
\newblock Evolution of hypersurfaces by their curvature in riemannian
  manifolds.
\newblock {\em Documenta Mathematica}, pages 349--360, 1998.

\bibitem{Jang:1979zz}
Pong~Soo Jang.
\newblock {Note on cosmic censorship}.
\newblock {\em Phys. Rev.}, D20:834--838, 1979.

\bibitem{Khuri:2016ulv}
Marcus Khuri and Naqing Xie.
\newblock {Inequalities Between Size, Mass, Angular Momentum, and Charge for
  Axisymmetric Bodies and the Formation of Trapped Surfaces}.
\newblock arXiv, gr-qc: 1610.04892, 2016.

\bibitem{Khuri:2015zla}
Marcus~A. Khuri.
\newblock {Existence of Black Holes Due to Concentration of Angular Momentum}.
\newblock {\em JHEP}, 06:188, 2015.

\bibitem{Khuri:2015xpa}
Marcus~A. Khuri.
\newblock {Inequalities Between Size and Charge for Bodies and the Existence of
  Black Holes Due to Concentration of Charge}.
\newblock {\em J. Math. Phys.}, 56(11):112503, 2015.

\bibitem{Kreiss-Ortiz-book}
Heinz-Otto Kreiss and Omar~Eduardo Ortiz.
\newblock {\em Introduction to Numerical Methods for Time Dependent
  Differential Equations}.
\newblock John Wiley \& Sons, Hoboken, NJ, first edition, 2014.

\bibitem{Malec:2002ki}
Edward Malec, Marc Mars, and Walter Simon.
\newblock {On the Penrose inequality for general horizons}.
\newblock {\em Phys. Rev. Lett.}, 88:121102, 2002.

\bibitem{Murchadha86b}
Niall~\'{O} Murchadha.
\newblock How large can a star be?
\newblock {\em Phys. Rev. Lett.}, 57(19):2466--2469, 1986.

\bibitem{0004-637X-757-1-55}
Feryal \"Ozel, Dimitrios Psaltis, Ramesh Narayan, and Antonio~Santos
  Villarreal.
\newblock On the mass distribution and birth masses of neutron stars.
\newblock {\em The Astrophysical Journal}, 757(1):55, 2012.

\bibitem{Reiris:2014tva}
Martin Reiris.
\newblock {On the shape of bodies in General Relativistic regimes}.
\newblock {\em Gen. Rel. Grav.}, 46:1777, 2014.

\bibitem{Schoen:2012nh}
Richard Schoen and Xin Zhou.
\newblock {Convexity of reduced energy and mass angular momentum inequalities}.
\newblock {\em Annales Henri Poincare}, 14:1747--1773, 2013.

\bibitem{SchoenYau1983}
S.-T Schoen~R., Yau.
\newblock {The Existence of a Black Hole Due to Condensation of Matter}.
\newblock {\em Communications in Mathematical Physics}, 90:575--579, 1983.

\bibitem{Senovilla:2007dw}
Jose M.~M. Senovilla.
\newblock {A Reformulation of the Hoop Conjecture}.
\newblock {\em Europhys. Lett.}, 81:20004, 2008.

\bibitem{Stergioulas:2003yp}
Nikolaos Stergioulas.
\newblock {Rotating Stars in Relativity}.
\newblock {\em Living Rev. Rel.}, 6:3, 2003.

\bibitem{Szabados04}
L{\'a}szl{\'o}~B. Szabados.
\newblock Quasi-local energy-momentum and angular momentum in {GR}: A review
  article.
\newblock {\em Living Rev. Relativity}, 7(4), 2004.
\newblock cited on 8 August 2005.

\bibitem{Unruh:1982ic}
W.~G. Unruh and Robert~M. Wald.
\newblock {Acceleration Radiation and Generalized Second Law of
  Thermodynamics}.
\newblock {\em Phys. Rev.}, D25:942--958, 1982.

\end{thebibliography}
\end{document}